\documentclass[twocolumn]{revtex4-2}

\usepackage{graphicx,indentfirst,csquotes,amssymb,amsthm,amsmath,mathtools}
\usepackage{xcolor,paralist,hyperref,titlesec,fancyhdr,etoolbox}
\newtheorem{theorem}{Theorem}[]
\newtheorem{definition}[theorem]{Definition}

\newtheorem{lemma}[theorem]{Lemma}

\hypersetup{ colorlinks=true, linkcolor=black, filecolor=black, urlcolor=black }

\usepackage[normalem]{ulem}
\newcommand{\llangle}{\langle\hspace{-1.5pt}\langle}
\newcommand{\rrangle}{\rangle\hspace{-1.5pt}\rangle}

\begin{document}

\title{Topological building blocks of nonequilibrium response}

\author{Sean Fancher}
\affiliation{Department of Biophysics, University of Michigan, Ann Arbor, MI 48109, USA}
\author{Jordan M. Horowitz}
\affiliation{Department of Biophysics, University of Michigan, Ann Arbor, MI 48109, USA}
\affiliation{Center for the Study of Complex Systems, University of Michigan, Ann Arbor, MI 48109, USA}
\affiliation{Department of Physics, University of Michigan, Ann Arbor, MI 48109, USA}




\begin{abstract}

Nonequilibrium systems can exploit energy to amplify their sensitivity to external stimuli, allowing them to be harnessed for a variety of functions in both engineered devices and living organisms.
An assortment of theoretical results capture different facets of this nonequilibrium amplification, including fluctuation-response relations as well as bounds and constraints that limit the potential behavior.
Here, we take a broader perspective, aiming for a systematic characterization of the full  range of possible response behaviors.
For a wide class of nonequilibrium dynamics,  we identify a collection of optimally sensitive models whose behavior is determined by the topology of the state space.
We further conjecture that every response can be written as a convex combination of these optimal models, thereby structuring the entire space of responses.
We use this geometric perspective to put sensitivity limits on nonmonotonic biochemical input-output functions, and identify all optimal kinetic schemes for the unordered binding of molecules among three sites.
\end{abstract}

\maketitle

Nonequilibrium systems are capable of a strikingly sensitive response to external stimuli.
Indeed, living systems offer a diverse array of examples where higher sensitivity correlates with improved function, such as more reliable biochemical signal transduction and sensing~\cite{Lan2012,Qian2012,Govern2014,Murugan2011,Sartori2013,Sartori2015}, higher fidelity molecular discrimination~\cite{Sartori2013b,Sartori2015b,Murugan2014,Wong2017,Morita2025}, increased responsiveness in biochemical switches~\cite{tu2008nonequilibrium,Mattingly2025}, and enhanced efficacy in transcriptional regulation~\cite{Bintu2005,Estrada2016,Tran2018,Park2019,Owen2023,Mahdavi2023}.
The broad relevance of these examples point to the need of a comprehensive theory of optimal nonequilibrium sensitivity.
Such a development would allow us to systematically explore mechanisms for sensitivity amplification and suggest unifying design principles that apply across systems.

In searching for general principles, it is natural to build a framework scaffolded on our equilibrium intuition gleaned from the fluctuation-dissipation theorem (FDT), which rationalizes response through fluctuations.
From this vantage point, analysis has focused on how the FDT is violated \cite{Harada2005} or in establishing nonequilibrium links between fluctuations and response.
This has been accomplished in direct analogy to the equilibrium FDT for specific equilibrium-like perturbations or initial preparations~\cite{Lubensky2010,Altaner2016}, or in general by connecting response to fluctuations in the sensitivity of the steady-state distribution \cite{Agarwal1972,Prost2009,Seifert2010}, in kinetic observables \cite{Speck2006,Baiesi2009,Chetrite2009,Seifert2010b,Chaudhuri2012,Aslyamov2025,Ptaszynski2026} or in force-fields \cite{Baldovin2022}.

An alternate route has emerged of late aimed at identifying nonequilibrium limits to response as a means to identify optimal system designs.
They have taken the form of  algebraic constraints~\cite{Aslyamov2024b,DalCengio2025,Floyd2025,Khodabandehlou2025,Bebon2026,Yang2026} and response bounds that offer fundamental trades-offs with model structure, thermodynamics~\cite{Owen2020,Owen2022,Owen2023,Martins2023,Chun2021,Gao2022,Gao2024,Chun2023,Ptaszynski2026} and kinetic properties \cite{Baiesi2011,Dechant2020,Ptaszynski2024,Aslyamov2024,Aslyamov2025,Kwon2025,Dechant2025,Zheng2025,Zheng2025b, Chun2026}.
From a broader perspective, though, the above predictions are delineating different facets of the space of all  responses, without capturing its entirety.
In this work, we take the first steps in characterizing this entire space for a class of systems driven arbitrarily far from equilibrium.
Specifically, we address the response of noisy systems that can be modeled via random jumps among a discrete collection of states, which encompasses numerous (bio)physical examples, including those above. 
For these systems, we identify a collection of 
maximally-sensitive models determined solely by the topology of the system's state space.
We further conjecture, based on analytic calculations and numerical exploration, that these optimal models are the only extreme points of the response space, implying that any model's response is a (convex) combination of them.
Together these observations suggest a potential geometric characterization of nonequilibrium response from which varied response behavior can be unified and systemized.

\section*{Shape of the response landscape}
\subsection*{Response in nonequilibrium steady states} As a modeling paradigm, we analyze a continuous-time Markov jump process on a graph $\mathcal{G}$ whose nodes $i=1,\ldots, N$ represent a collection of discrete physical configurations or states. 
Directed edges $e$ between pairs of states signify allowable transitions with (kinetic) rates $W_{e}$. 
In such a model, the probability $p_i(t)$ for a system to be in state $i$ at time $t$ evolves according to the master equation~\cite{VanKampen}
\begin{equation}\label{eq:master}
\partial_t p_i(t)=\sum_{j=1}^N W_{ij}p_j(t),
\end{equation}
with off-diagonal terms $W_{ij} = \sum_{e\in\mathcal{E}_{ij}} W_{e}$, where $\mathcal{E}_{ij}$ is the set of all $j\to i$ edges, and diagonal entries $W_{ii}=-\sum_{j\neq i}W_{ji}$ fixed by probability conservation.
We will assume every pair of states can be reached by a sequence of transitions --- that is the graph ${\mathcal G}$ is strongly connected --- so that there exists a unique normalized steady-state distribution $\pi_i$ given as the solution of $\sum_j W_{ij}\pi_j =0$.
With knowledge of $\pi_i$, we can predict the steady-state averages of physical observables $Q_j$ defined as $\langle Q\rangle = \sum_j Q_j \pi_j$, which are often more easily accessed experimentally.

Our central aim is to understand how the topology of ${\mathcal G}$ shapes the response of the steady state to coordinated changes in the kinetic rates. To this end, it is useful to parameterize the rates via an external control parameter $\lambda$ that we vary to perturb the dynamics, $W_e(\lambda)$. 
This parameter is rather arbitrary; potentially ranging from macroscopic, like an external electric field, to microscopic, as in a kinetic barrier. 
Once this parameterization is fixed, the scale of the perturbation is determined by the sensitivity of the rates to $\lambda$, 
\begin{equation}\label{eq:d}
d_e(\lambda) = \partial_\lambda \ln W_{e}(\lambda).
\end{equation}
Consider, as an example, a thermally-activated transition $W_{e}\propto e^{-\beta(B_{ij}-E_j)}$ in which $e$ is directed $\{j\to i\}$ and $B_{ij}-E_j$ is an activation energy at inverse-temperature $\beta$.
We can introduce externally-controlled changes to the energy landscape by coupling $\lambda$ to a conjugate coordinate $V_j$ via $E_j \to E_j+\lambda V_j$.
In this scenario, the sensitivity reduces to the conjugate coordinate, $d_{e} = V_j$, familiar from equilibrium response theory.

With this parameterization, the steady state becomes an implicit function of the external parameter $\pi_j(\lambda)$ via the kinetic rates, and variations in $\lambda$ translate into the steady-state response
\begin{equation}\label{eq:Qder}
\partial_\lambda \langle Q\rangle=\sum_{j=1}^N Q_j \partial_\lambda\pi_j =\sum_{j=1}^N Q_{j}\pi_{j}\partial_\lambda \ln\pi_{j},
\end{equation}
assuming, as we do for simplicity of presentation, that the physical observables are independent of $\lambda$.
Thus, by searching over all the transition rates $W_e$ with the same sensitivities $d_e$, we can explore the potential values of the steady-state sensitivities $\partial_\lambda \pi_j$, some of which will optimize the observable response $\partial_\lambda \langle Q\rangle$.
To characterize these optima we will find it convenient to consider a more constrained problem, which can be relaxed later, where we hold the steady state $\pi_j$ fixed as well.

Putting together all these observations motivates us to view \eqref{eq:Qder} as a linear objective function to be optimized over the vector of steady-state sensitivities $\partial_\lambda \ln \pi_j$.
Crucial to the solution of this problem is a characterization of the response space ${\mathcal R}$, defined here as the set of all sensitivities $\partial_\lambda \ln \pi_j$ achievable for some collection of transition rates, all with fixed $\pi_j$ and $d_e$.
By considering the sensitivities $\partial_\lambda \ln \pi_j$ as coordinates in an $N$-dimensional vector space, we have $\mathcal{R}\subseteq\mathbb{R}^{N}$.
However, normalization of the steady state  requires that $\langle\partial_{\lambda}\ln\pi_{j}\rangle=0$, meaning $\mathcal{R}$ is contained within an $N-1$ dimensional hyperplane of $\mathbb{R}^{N}$.
This is still an unbounded and unstructured characterization of ${\mathcal R}$.
Our main contribution will be the observation that ${\mathcal R}$ is actually a finite subset of a polytope specified by the topology of ${\mathcal G}$.

To link response theory to topology we exploit the Matrix Tree Theorem, which provides a powerful graph-theoretic representation of the steady-state distribution.
This theorem~\cite{schnakenberg1976,Hill} and its extensions~\cite{Caplan1982, Chaiken1982, Chebotarev2002} have helped us rationalize not only the steady-state distribution~\cite{Maes2012,Liang2022, Cetiner2022,Cetiner2023} but also its response~\cite{Polettini2017,Khodabandehlou2022,Floyd2025,Martins2023,Owen2022,Owen2023}. 
What follows builds on these observations and extends them, explaining how ${\mathcal G}$ also structures the entire space of potential responses ${\mathcal R}$.
To see this, we need to take a moment to introduce some concepts from graph theory that capture this topology.

\subsection*{Topology of the state graph}
For our analysis the main players are rooted spanning trees $T_r$~\cite{schnakenberg1976,Hill}: connected subgraphs of $\mathcal{G}$ that contain every node, no cycles, and each edge is oriented so that it points towards the root $r$, as exemplified in Fig.~\ref{fig:triangles}.
They are linked to the kinetics by assigning each tree a weight $\omega(T_r)=\prod_{e\in T_r} W_e$, as the product of the rates along the edges of the tree.
From the trees, we construct a new object by assembling a list of $N$ rooted spanning trees with each tree rooted at a different node, $\sigma =\{T^\sigma_1,\ldots, T^\sigma_N\}$, into a collection we term a {\it family of rooted spanning trees} or simply a tree family, with superscript labeling the tree family.
Similarly to the tree weights, tree families are assigned a weight $\omega(\sigma)=\prod_{k=1}^N \omega(T_{k}^\sigma)$, as the product of the weights of the trees.

\subsection*{Geometric characterization of response}
Tree families turn out to be the linchpin for identifying the geometric structure of response, once we couple them to the perturbations through a collection of observables  we call {\it topological scaling factors}
\begin{equation}
G_j^\sigma = \partial_{\lambda}\ln\left(\omega\left(T_{j}^{\sigma}\right)\right) = \sum_{e\in T_j^\sigma} d_e.
\end{equation}
This notation allows for our first result.
By directly differentiating the Matrix Tree Theorem expression for the steady-state distribution, the response can be organized as the convex combination
\begin{equation}\label{eq:convex_full}
\partial_\lambda\ln\pi_j = \sum_{\sigma}\frac{\omega\left(\sigma\right)}{\mathcal{Z}}\left(G_j^\sigma-\left\langle G^\sigma\right\rangle\right),
\end{equation}
where $\mathcal{Z}=\sum_{\sigma'}\omega(\sigma')$ is a normalization factor (derivation outlined in the Materials and Methods with details in the SI `convex representation of response').
It is worth noting that the convex weights $\omega(\sigma)/{\mathcal Z}$ are invariant under a site-dependent rescaling of the rates, $W_{e}\to W_{e}/\tau_{j}$ for all $e$ exiting node $j$.
This follows by noting that every tree family has exactly $N-1$ edges leaving a state (one edge for each tree in which that state is not the root), so that every tree family weight rescales identically, $\omega(\sigma)\to\omega(\sigma)/\prod_{j=1}^N \tau^{N-1}_j$, leaving the ratios $\omega(\sigma)/{\mathcal Z}$ unchanged.
Conversely, such a transformation does affect $\pi_{j}$, which increases with $\tau_{j}$.
Thus, as we vary the rates to explore all the potential sensitivities, we can freely readjust the scales $\tau_j$ to keep the steady-state distribution fixed. 
 
Equation (\ref{eq:convex_full}) immediately implies that ${\mathcal R}$ lies inside a polytope whose vertices are the tree families' (mean-shifted) topological scaling factors $v^\sigma_j=G_{j}^{\sigma}-\langle G^{\sigma}\rangle$, which represent potential optima.
Furthermore, each vertex satisfies the zero-mean normalization condition, $\langle v^\sigma\rangle=0$, meaning that this entire polytope  lies in the same hyperplane of $\mathbb{R}^{N}$ as $\mathcal{R}$.
So as long as the topological scaling factors are finite, the responses will be contained within a finite region with a shape determined by the topology of $\mathcal{G}$ (via the tree families) and the sensitivity of the kinetics (via the $d_e$).

While \eqref{eq:convex_full} offers a geometric characterization of ${\mathcal R}$, it does suffer from a deficiency.
The polytope is too large.
As we will show, not every tree family (i.e., $G_j^\sigma-\langle G^\sigma\rangle$) is contained within ${\mathcal R}$; that is for some tree families there is no combination of rates such that $\omega(\sigma)/{\mathcal Z}=1$ and  $\partial_\lambda \ln\pi_j =(G_j^\sigma-\langle G^\sigma\rangle)$.
Such tree families are thus not \emph{physically realizable}.
Critically though, those physically-realizable tree families that are in ${\mathcal R}$ share a specific topological property that we now elucidate.
\begin{figure}
\centering
\includegraphics[width=.95\linewidth]{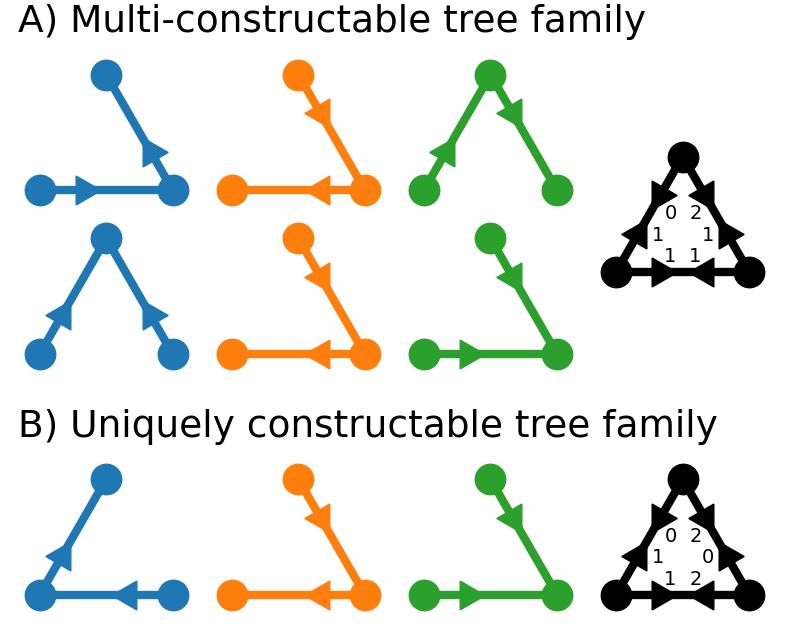}
\caption{Some example tree families of the 3 state graph. A) Two distinct tree families with trees rooted at the top (blue), left (orange), and right (green) states. The number of times each directed edge is used is given by the black graph to the right. Since both sets use the same collection of edges with the same multiplicities, they can each be reconstructed into the other and are thus mulit-constructable. B) Another possible tree family with its edge multiplicities. In this case, no other construction is possible, thus making this set uniquely constructable.}
\label{fig:triangles}
\end{figure}

Consider the two tree families depicted in Fig.\ \ref{fig:triangles}A.
We label these sets as {\it multi-constructable} due to the fact that despite containing different trees, they can be constructed from the same constituent edges with the same multiplicities, as shown in the black graph on the right.
This directly implies that they have the same weight $\omega(\sigma)$.
However, they are still distinct tree families and thus have distinct topological scaling factors.
The result is that both their contributions to the sum in \eqref{eq:convex_full} are bounded above by $\omega(\sigma)/\mathcal{Z}\le 1/2$.
Thus there is no set of rates realizing these tree families.

Contrast this with the tree family depicted in Fig.\ \ref{fig:triangles}B.
In this case, there is no other tree family that can be constructed from the same constituent edges.
It is an example of what we call a {\it uniquely-constructable} family, which we label using $\alpha$ instead of $\sigma$ as such families will play a key role in what follows.
Since they do not face the same bounding problem as the multi-constructable families, their $\omega(\alpha)/{\mathcal Z}$ can feasibly obtain a unit value, potentially making them physically realizable.
In fact, we have proven that all uniquely-constructable families are always physically realizable, while multi-constructable ones are never physically realizable (see SI `every physical-realizable tree family is uniquely constructable').
Put another way, none of the multi-constructable families can be optimal, while the uniquely-constructable ones are a subset of the potential extrema.

This distinction is crucial as it imposes the condition that all the polytope vertices $v^\sigma_j=G_j^\sigma-\langle G^\sigma\rangle$ corresponding to uniquely-constructable families must be contained within the closure of $\mathcal{R}$.
Given this, one may surmise that the responses can be written as a convex combination of \emph{only} the uniquely-constructible vertices,
\begin{equation}\label{eq:convex}
\partial_{\lambda}\ln\pi_{j} = \sum_{\alpha} \theta^{\alpha}\left(\lambda\right)\left(G^{\alpha}_{j}-\left\langle G^{\alpha}\right\rangle\right),
\end{equation}
where the positive coefficients $\theta^\alpha(\lambda)\ge 0$ are normalized $\sum_\alpha \theta^\alpha(\lambda)=1$, implicitly functions of $\lambda$ via the rates, and as linear combinations of the $\omega(\sigma)/{\mathcal Z}$ can be varied independently of the steady state.

However, \eqref{eq:convex} is not guaranteed to be true given what we have discussed.
$\mathcal{R}$ may well have curved surfaces that bulge outside of the convex hull of the uniquely-constructable families ($v^\alpha_j$) while still being contained within the convex hull of all tree families ($v^\sigma_j$).
Here, though, we conjecture that \eqref{eq:convex} is indeed valid.
Our conjecture is based on the aforementioned fact that the multi-constructable families are never physically realizable as well as direct confirmation for the 3 state graph (see SI `convex parameters for the 3-state graph').
Furthermore, we have also explored this conjecture numerically and verified it for many example networks, such as seen in Fig.~\ref{fig:convex_ver}.
\begin{figure}[tb]
\centering
\includegraphics[width=.95\linewidth]{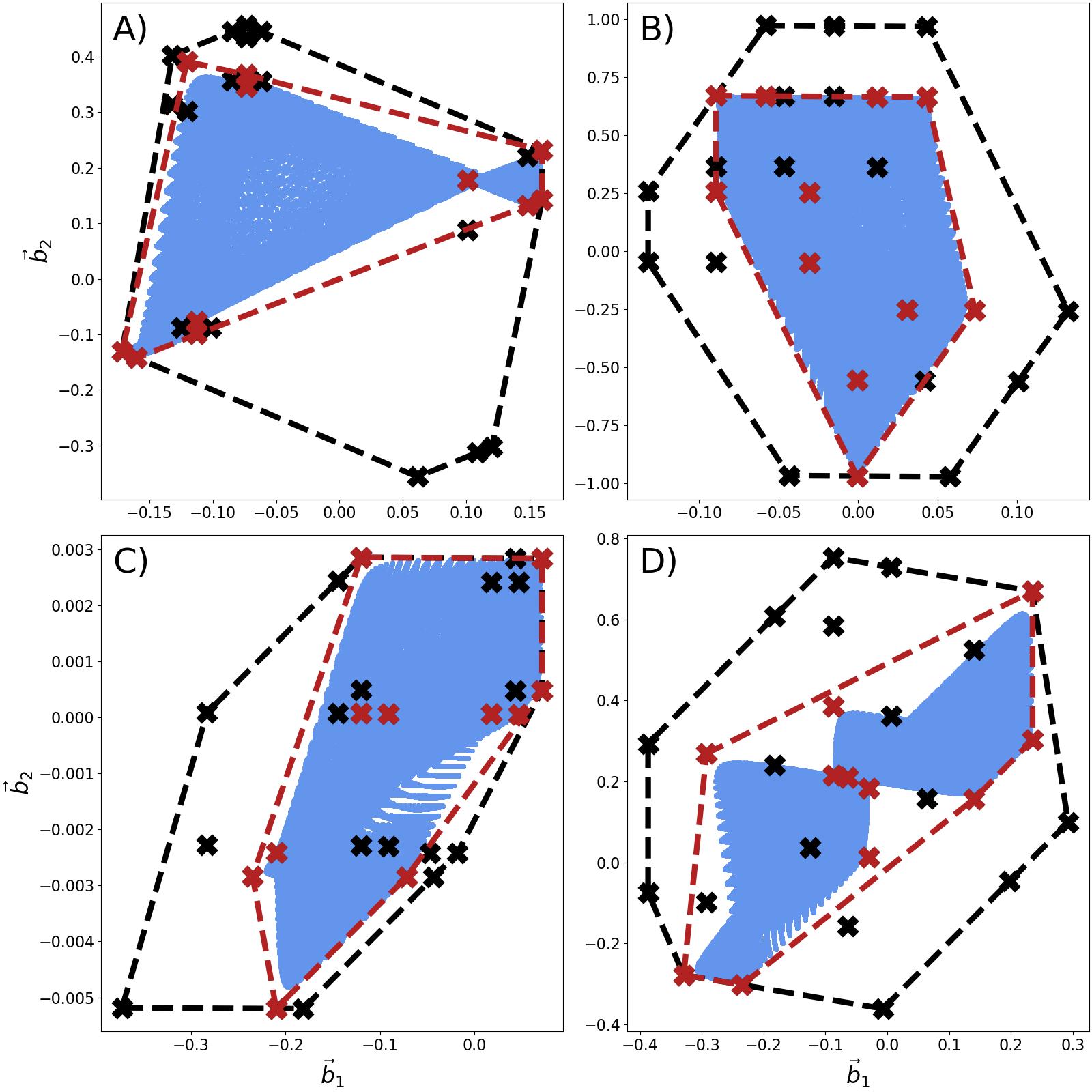}
\caption{Numeric exploration of response in the 3 state graph. The response of this system exists within a 2-dimensional subspace of the 3-dimensional space defined by $\partial_\lambda\ln\pi_j$. Axes represent an orthonormal basis of this subspace for each parameter set. The position of the tree family vertices ($v^{\sigma}$) are given by the $\times$ markers with multi-constructable tree families in black and uniquely-constructable ones in red. The convex hulls of these two collections are given by their correspondingly colored dashed lines. Blue points provide numerically calculated response values from a sweep over the space of possible rates under fixed $\pi$ and $d$ values, as described in the Materials and Methods. Critically, each such realization lies within the drawn convex hull of the uniquely-constructable points, though the hull is not always completely filled.}
\label{fig:convex_ver}
\end{figure}
Equation (\ref{eq:convex}) simply represents an extrapolation of these results to all possible graphs.

The implication for the response of arbitrary observables is readily obtained by substituting \eqref{eq:convex} into \eqref{eq:Qder},
\begin{equation}\label{eq:Qconvex}
\partial_\lambda\langle Q\rangle = \sum_\alpha \theta^\alpha \llangle QG^\alpha \rrangle,
\end{equation}
in terms of the steady-state covariances  $\llangle QG^\alpha\rrangle=\langle Q G^\alpha\rangle-\langle Q\rangle\langle G^\alpha\rangle$.
Equation (\ref{eq:Qconvex}) links the response to fluctuations in far-from-equilibrium scenarios, naturally drawing connections to the FDT and its nonequilibrium variations.
The notable distinction from previously explored results is that the ``conjugate coordinates'' $G^\alpha$ are merely determined  by the interplay between the external perturbation and the state-space topology.
Thus, the $G^\alpha$ can be predicted without solving for the steady-state distribution or inferring other dynamical quantities.
The cost we pay is that without a program for determining the convex parameters $\theta^\alpha$, we cannot predict the response directly.
Even still, knowledge of such a relationship provides new insights, as we now explore. 
\section*{Response in optimal models}
Equation (\ref{eq:Qconvex}) suggests that the the uniquely-constructable families act as a kind of archetypical primitives out of which all responses are built.
Understanding their structure can then shed light on the breadth of potential response behavior.

\subsection*{Sensitivity bounds}
With this geometry we can place limits on the range of responses for any arbitrary observable, $\partial_\lambda \langle Q\rangle$.
Indeed, \eqref{eq:Qconvex} is extremized when for some tree family $\theta^{\alpha}=1$ with the remaining $\theta^{\alpha}$ zero.
This implies that the maximum and minimum possible responses are
\begin{equation}\label{eq:ineq}
\min_\alpha \llangle QG^\alpha\rrangle \le  \partial_\lambda \langle Q\rangle \le \max_\alpha \llangle QG^\alpha\rrangle.
\end{equation}
This fluctuation bound represents the tightest possible limit to the sensitivity among all models on ${\mathcal G}$ with the same steady-state distribution and topological scaling factors.
From it weaker bounds can readily be derived in terms of observable fluctuations (e.g., $\langle Q\rangle$, $\llangle Q^2\rrangle$, $\langle G^\alpha\rangle$...) using covariance inequalities like those discussed in Ref.~\cite{Martins2023}.
As such, it significantly generalizes the support bound~\cite{Owen2023}, which is restricted to uniform perturbations (all $d_e$ equal or zero) and does not account for correlated responses across different states.
In Fig.\ \ref{fig:normsense}, we validate this numerically by plotting the normalized sensitivity
\begin{equation}\label{eq:normsens}
\hat{S}\left(\lambda\right) = \frac{2\partial_{\lambda}\langle Q\rangle-\max_\alpha \llangle QG^\alpha\rrangle-\min_\alpha \llangle QG^\alpha\rrangle}{\max_\alpha \llangle QG^\alpha\rrangle-\min_\alpha \llangle QG^\alpha\rrangle}\in\lbrack-1,1\rbrack.
\end{equation}
for several example systems, providing further  evidence for our conjecture (\eqref{eq:convex}).
\begin{figure}[tb]
\centering
\includegraphics[width=.95\linewidth]{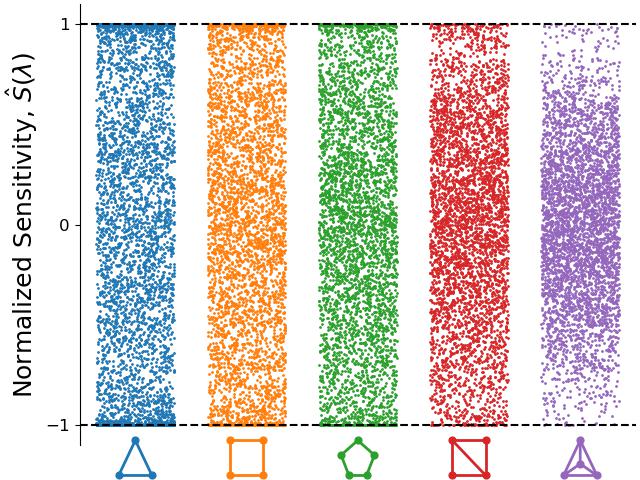}
\caption{Numeric exploration of sensitivity in several randomized systems. For each of the depicted networks, values of $Q$, $W$, and $d$ are all randomly determined from standard uniform, log-normal, and normal distributions, respectively. The corresponding $\partial_\lambda \langle Q\rangle$, $\min_\alpha \llangle QG^\alpha\rrangle$, and $\max_\alpha \llangle QG^\alpha\rrangle$ are then calculated, and the normalized sensitivity is used so that $\hat{S}(\lambda)=+1$ corresponds to $\partial_\lambda \langle Q\rangle = \max_\alpha \llangle QG^\alpha\rrangle$ and $\hat{S}(\lambda)=-1$ to $\partial_\lambda \langle Q\rangle = \min_\alpha \llangle QG^\alpha\rrangle$. This process of drawing new values of $Q$, $W$, and $d$ and recalculating $\hat{S}(\lambda)$ is repeated a total of 4096 times for each network. This data shows the validity of the bounds presented in \eqref{eq:ineq}.}
\label{fig:normsense}
\end{figure}

Equation~(\ref{eq:ineq}) teaches us that to design an optimally-responsive model, we need to tune the rates until we realize one of the uniquely-constructable tree families.
In the SI `every physical-realizable tree family is uniquely constructable', we show this is always possible and occurs when the rates conspire so that the weights of all the $T^\alpha_j\in \alpha$ are by far the largest, i.e, for every $j$ we have $\omega(T_j^\alpha)\gg\omega(T_j)$ for all $T_j \neq T^\alpha_j$.

\subsection*{Parameterizing the steady-state}
Further insight can be gained by focusing on a special class of models where the rate-sensitivities $d_e$ are fixed and constant.
This is a common motif in biochemical models where the perturbations of interest are changes in the log-concentration of an input molecule, $\lambda=\ln c$~\cite{tu2008nonequilibrium,Mattingly2025,Bintu2005,Estrada2016,Tran2018,Park2019,Owen2023,Mahdavi2023}.
Assuming mass-action kinetics, the rates are either constant with respect to $c$ ($d_{e}=0$) or proportional to the concentration with $W_e=k_e c$ and $k_e$ independent of $c$ ($d_e=\partial_{\ln c}\ln W_e=1$).
For such models, we can easily work out the optimal steady-state distribution for all parameter values and explore its structure.

To this end, imagine the rates have be tuned so that one of the uniquely-constructable families, $\alpha$, has been physically realized over some range of $\lambda$ values, $\lambda\in\mathcal{L}$, which we denote as the \textit{realization range}.
In such an optimal scenario, we can set $\theta^{\alpha}(\lambda)=1$ over $\mathcal{L}$ (see Materials and Methods `uniquely-constructable steady state'), and \eqref{eq:convex} becomes a differential equation for the optimal steady-state distribution with solution
\begin{equation}\label{eq:uniqConstDistr}
\pi_j(\lambda) = \frac{1}{\mathcal N}e^{\lambda G_j^{\alpha}}\pi_j(0),
\end{equation}
referenced against a fixed, though arbitrary, initial steady-state distribution $\pi_j(0)$ at $\lambda=0$, with normalization constant ${\mathcal N} = \sum_j e^{\lambda G^\alpha_j}\pi_j(0)$.
Remarkably, the locus of steady states traced out by varying $\lambda$ within $\mathcal{L}$ can be obtained through a simple exponential reweighting via the topological scaling factor $G_j^\alpha$.
While this is not a general solution for the steady state, it does provide a nontrivial relationship between different steady-state distributions.
In particular, we can use it to predict the shape of the response curve for $\lambda\in\mathcal{L}$
\begin{equation}\label{eq:optObs}
\langle Q\rangle_\lambda =  \frac{1}{\mathcal N}\sum_j Q_j e^{\lambda G_j^{\alpha}}\pi_j(0).
\end{equation}

This response curve can display a variety of complicated structures.
To contextualize them, we turn to biochemistry where Hill functions serve as the basic phenomenological organizing equations \cite{hill1910possible,sorribas2007cooperativity,stefan2013cooperative}:
\begin{equation}\label{eq:hilldef}
h^{+}\left(c\right) = \frac{c^{n}}{K^{n}+c^{n}}; \quad h^{-}\left(c\right) = \frac{K^{n}}{K^{n}+c^{n}}.
\end{equation}
Here, the Hill coefficient $n$ measures the senstivity,
\begin{equation}\label{eq:hillextract}
n=\frac{\partial_{\ln c}h^{\pm}\left(c\right)}{h^{\pm}\left(c\right)\left(1-h^{\pm}\left(c\right)\right)},
\end{equation}
$K$ is the half-maximum point, and the $\pm$ superscripts denote whether the Hill function is increasing or decreasing with the input variable, $c$. Note that \eqref{eq:hillextract} produces a negative Hill coefficient for decreasing Hill functions.

By inspection of \eqref{eq:optObs}, we can arrange a linear combination of Hill functions by tuning the steady state distribution $\pi_j(0)$ to be nonzero on only two sets of states: $j\in{\mathcal S}_>$ with  $G_{j}^{\alpha}=G_>$ and $l\in{\mathcal S}_<$ with $G_{l}^{\alpha}=G_<$. 
Under these conditions, the response curve only depends on $\Delta G=G_>-G_< > 0$ as
\begin{equation}\label{eq:hillform}
\langle Q\rangle_\lambda = \frac{\langle Q\rangle_<+\langle Q\rangle_>e^{\lambda\Delta G}}{\langle 1\rangle_<+\langle 1\rangle_>e^{\lambda\Delta G}},
\end{equation}
where $\langle Q\rangle_>=\sum_{j\in {\mathcal S}_>}Q_j\pi_j(0)$ and similar for $\langle Q\rangle_<$.
Equation (\ref{eq:hillform}) is clearly a linear sum of an increasing and decreasing Hill function, both of which have
\begin{equation}\label{eq:DelG}
\quad n = \Delta G.
\end{equation}

This simple phenomenon is exemplified in Fig.~\ref{fig:hill}, where we consider the three state network once again.
\begin{figure}[!tb]
\centering
\includegraphics[width=.85\linewidth]{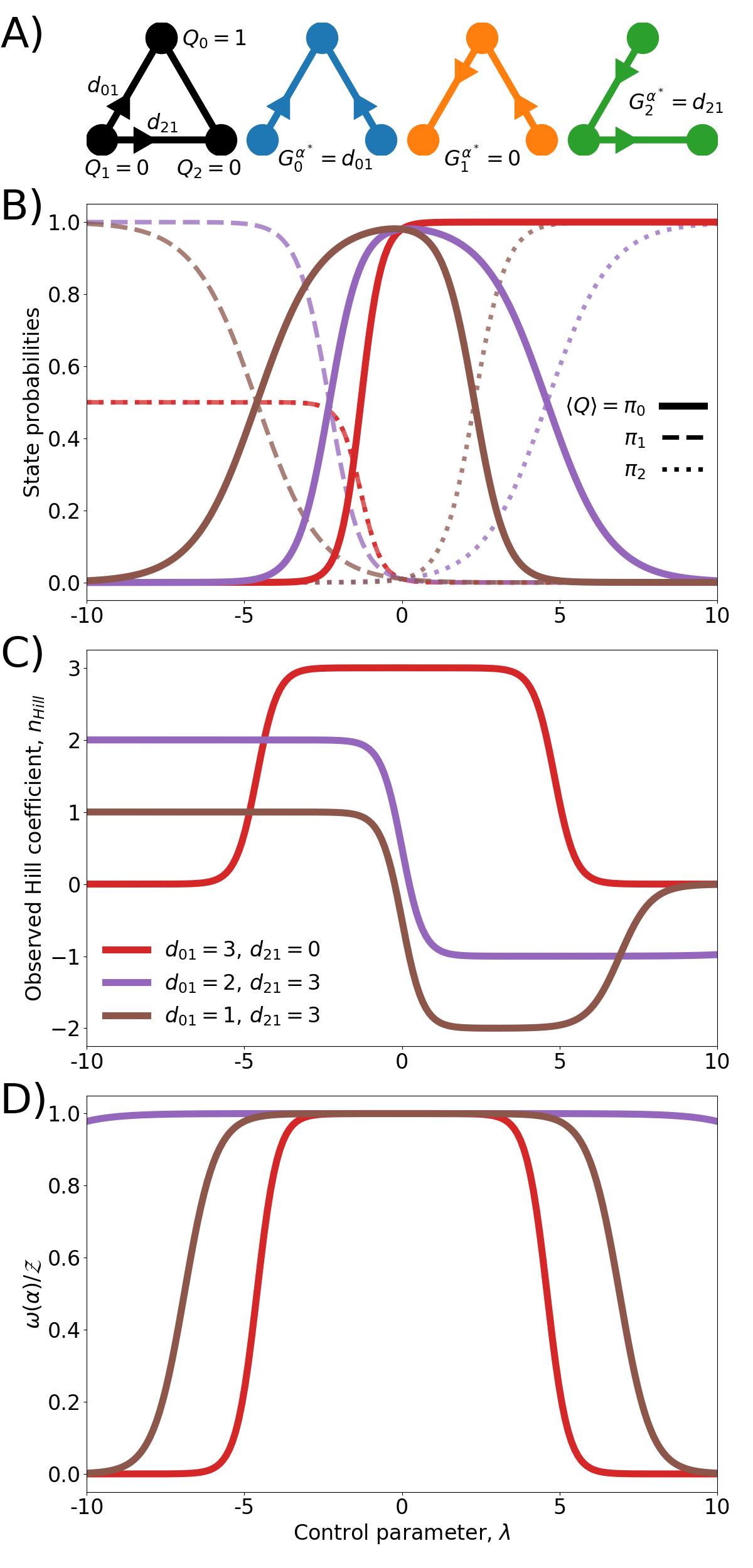}
\caption{A) The three state network with displayed values of $Q$ and $d$ is given base rates $W_{01}=1$, $W_{10}=10^{-2}$, $W_{02}=W_{21}=10^{-6}$, $W_{12}=10^{-12}$, and $W_{20}=10^{-14}$ so that the depicted uniquely constructable tree family is physically realized. B) The steady state probabilities of each state (and equivalently $\langle Q\rangle$) showcase the characteristic features of single (red) and double (purple and brown) sided Hill functions. C) The Hill coefficient is calculated via \eqref{eq:hillextract} for each system as a function of $\lambda$. The values in the transition regions are seen to perfectly match those predicted by \eqref{eq:DelG} with the signs dictating ramp up and ramp down regimes respectively. D) The realization range ($\mathcal{L}$) for each curve is seen to extend over the transition regions of interest ($\mathcal{L}\approx\lbrack -4,4\rbrack$ for the red curve, $\approx\lbrack -10,10\rbrack$ for the purple, and $\approx\lbrack -6,6\rbrack$ for the brown).}
\label{fig:hill}
\end{figure}
This time the rates are chosen so as to physically realize the uniquely-constructable family depicted in Fig.\ \ref{fig:hill}A with the realization range (Fig.\ \ref{fig:hill}D) extending over the regions in which probability is seen to transfer between states in a Hill-like fashion.
The edges leaving the left state are taken to vary with $\lambda$ as $W_{01}\propto\text{exp}(\lambda d_{01})$ and $W_{21}\propto\text{exp}(\lambda d_{21})$ while the remaining rates are held fixed. 
For the particular system depicted in Fig.\ \ref{fig:hill}, only the top state has $Q=1$ while the left and right states have $Q=0$, which in turn allows \eqref{eq:hillform} to simplify down to a single Hill function. 

For different choices of rate-sensitivities $d_e$ we can exhibit different behaviors.
Figure \ref{fig:hill}B shows the resulting probabilities accompanied in Fig.\ \ref{fig:hill}C by the observed Hill coefficient  (cf. \eqref{eq:hillextract})
\begin{equation}
n_{\rm Hill}=\frac{\partial_{\lambda}\langle Q\rangle_\lambda}{\langle Q\rangle_\lambda\left(1-\langle Q\rangle_\lambda\right)}.
\end{equation}
In the case of the red curves, $\langle Q\rangle$ is a Hill function that shifts the probability from states 1 and 2 ($G_{<}=0$) to state 0 ($G_{>}=3$) with observed Hill coefficient $\Delta G=3$ in the transition region. 
The purple and brown curves complicate the matter slightly by containing two distinct transitions each, one rising ($\lambda<0$, $\pi_{2}\approx 0$) and one falling ($\lambda>0$, $\pi_{1}\approx 0$). 
In both of these cases, each state has a distinct $G_{j}^{\alpha}$ value, but the general form of \eqref{eq:hillform} is maintained on either side of $\lambda=0$ by having one of the states suffer a near zero probability and thus leave only the remaining two to significantly contribute to \eqref{eq:optObs}.
Specifically, the ramp up side moves probability from state 1 ($G_{<}=0$) to state 0 ($G_{>}=d_{01}$) while the ramp down side moves probability from state 0 ($G_{<}=d_{01}$) to state 2 ($G_{>}=d_{21}=3$). In this way, we see that such a system exhibits a tradeoff between the sensitivities on either side; the faster it ramps up (larger $\Delta G=d_{01}-0$), the slower it ramps down (smaller $\Delta G=d_{21}-d_{01}$). This highlights the more general result that in any system in which a particular tree family has been physically realized, any observed Hill coefficient is always bound by the maximum value of $\Delta G$ across the whole tree family. 
Similarly, any perturbation that shifts the probability between more than two states must satisfy that the sum over the observed Hill coefficients is similarly bound by the maximum $\Delta G$.

\subsection*{General three site model}
The previous example showcased how knowledge of a particular uniquely-constructable family lends insight into the potential structures of a system's response.
However, in systems of even moderate size the number of uniquely-constructable families can be quite large.
Here, we show how to conceptually explore these possibilities by utilizing other known properties of the state graph and the perturbation itself.

Our motivation here is a variety of biophysical processes relying on the unordered binding of indistinguishable particles to three sites.
These could be three transcription factor binding sites on a gene enhancer, or a triplet of ligand-binding sensory receptors on a cell's membrane.
As each site can be bound ($1$) or empty ($0$), each configuration is represented by a triplet of binary numbers (e.g., $(000)$, $(100)$, $(101)$ etc.), with the full state space depicted in Fig.~\ref{fig:3binding}.
\begin{figure}[t]
\centering
\includegraphics[width=.85\linewidth]{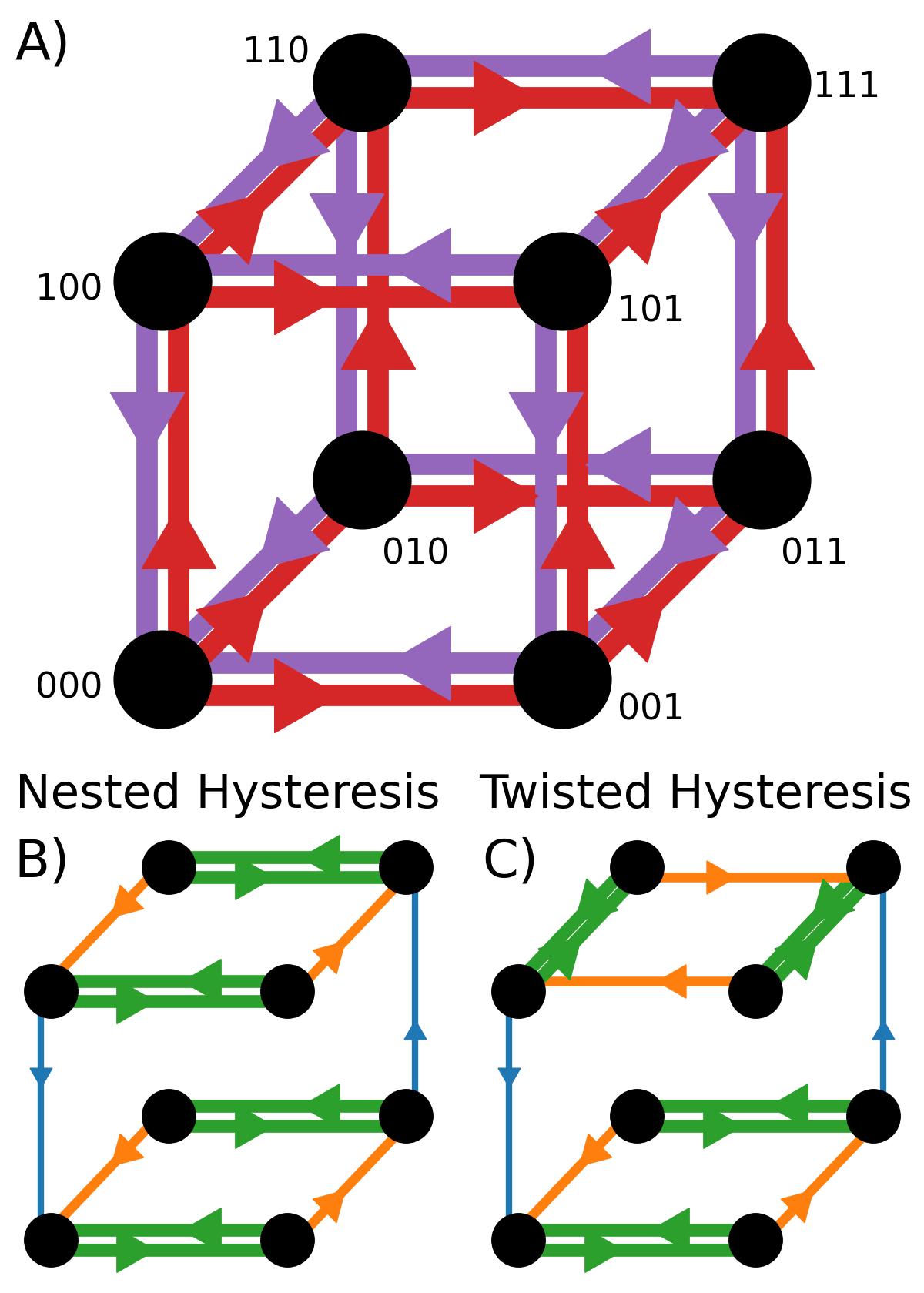}
\caption{The binary triplet model. A) The eight possible binary triplet states with all allowable transitions. Red edges represent binding edges ($d_{e}=1$) while purple edges represent unbinding edges ($d_{e}=0$). B) The Nested Hysteresis model. The green left-right edges hold the highest rates, the orange front-back edges are of moderate magnitude, and the blue top-bottom edges have the slowest transitions. This model is capable of achieving maximum sensitivity. C) The Twisted Hysteresis model. Rates are color coded in the same way as (B), but the top portion of the graph has been twisted. This model is also capable of achieving maximum sensitivity. See SI for full breakdown of the corresponding tree families.}
\label{fig:3binding}
\end{figure}
Here, the binding rates are commonly taken to be sensitive to the ligand concentration, $W_e\propto c\propto e^\lambda$, our control parameter, while the unbinding rates are not.
Then the most compelling variable of interest is often the fraction of bound sites, $Q_j\in\{0,1/3,2/3,1\}$.

By numerically exploring the space of all possible spanning trees, we can show that there exists 384 possible trees for this system.
Combining this with the $2^3=8$ configuration states yields $384^8\approx 10^{20}$ distinct tree families.
To find the uniquely-constructable ones that optimize the sensitivity, \eqref{eq:ineq}, we need to maximize the covariance $\llangle QG^\alpha\rrangle$.
We have found that utilizing the Cauchy-Schwarz inequality can aid in the search:
\begin{equation}
\llangle QG^\alpha\rrangle \le \sqrt{\llangle Q^{2}\rrangle\llangle \left(G^\alpha\right)^{2}\rrangle}.
\label{eq:covineq}
\end{equation}
Saturating this bound requires that $Q$ and $G^\alpha$ must vary in concert.
Maximizing it requires maximizing the variances of both, which we do by first searching over all $\pi_j$ and then identifying the optimal $G^\alpha_j$.

The simplest task is to maximize $\llangle Q^{2}\rrangle$ over all $\pi$. 
Only the fully unbound (000) state achieves the minimum $Q=0$ while only the fully bound (111) state achieves the maximum $Q=1$. 
The maximum variance $\llangle Q^{2}\rrangle=1/4$ is thus realized when these two extreme states are given equal and maximal weight, $\pi_{(000)}=\pi_{(111)}=1/2$.

Under this distribution, the variance of $G^\alpha$ can also be maximized by ensuring it has a minimum possible value in the fully unbound state and maximum possible value in the fully bound state. 
To generate this arrangement, we note with our choice of $d_e$, each $G_j^\alpha$ is merely the number of binding edges that exist within the $j$-rooted tree $T_j^\alpha$. 
Beginning with the fully unbound state, the minimum possible value of $G_{(000)}^\alpha$ occurs when $T_{(000)}^\alpha$ has no binding edges. 
Similarly, the maximum possible value of $G_{(111)}^\alpha$ occurs when $T_{(111)}^\alpha$ contains only binding edges. 
If both of these conditions can be satisfied simultaneously, then of $G_{(000)}^\alpha=0$ from the lack of any binding edges and $G_{(111)}^\alpha=7$ from the 7 edges that exist in any spanning tree.
The resulting variance is $\llangle \left(G^\alpha\right)^{2}\rrangle=49/4$.

Of the 384 possible spanning trees, only 24 of them contain no binding edges when rooted at the fully unbound state while another 24 contain only binding edges when rooted at the fully bound state (see SI). This space of $24\times 24$ possibilities can be further reduced by enforcing unique constructability. 
Indeed, upon scrutinizing all $24^{2}$ possible pairs of these trees, one finds that only 12 cannot be reconstructed into a different pair of trees.
Thus, the only uniquely-constructable familes with our extreme values ($G_{(000)}^\alpha=0$ and $G_{(111)}^\alpha=7$), must contain one of these 12 pairs.
The task then becomes to fill out the remains spanning trees rooted at the remaining states in such a way as to maintain unique constructability. 
As it turns out, each of these 12 uniquely constructable pairs can be expanded into exactly one full uniquely-constructable  tree family, thus producing 12 possible solutions that each obtain the maximum possible sensitivity of $\partial_{\lambda}\langle Q\rangle=\sqrt{\llangle Q^{2}\rrangle\llangle \left(G^\alpha\right)^{2}\rrangle}=7/4$.

These 12 solutions can be split into two distinct groups. The first set of 6 are simply the nested hysteresis scheme identified in Ref.\ \cite{Owen2023} and all of its rotations and reflections. The second set of 6 is equivalent to nested hysteresis, but with the binding order of the first two sites switched when the third becomes bound, plus its the rotations and reflections.
From this, we conclude that nested hysteresis and its twisted counterpart not only represent some of the possible optima, as discussed in Ref.\ \cite{Owen2023}, but are in actuality the only possible solutions.

Our results also allow us to address the minimum sensitivity given by the lower bound in \eqref{eq:ineq}.
We have already chosen $\pi$ such that the extreme states host all of the probability, so now we must invert our previous procedure and find the uniquely constructable tree family that has the highest possible number of binding edges in $T_{(000)}$ and the lowest possible number of binding edges in $T_{(111)}$. 
While it is possible to construct a $T_{(000)}$ that has as many as 4 binding edges and a $T_{(111)}$ that has as few as 3, no such pair of trees are uniquely constructable. 
In fact, by scanning over all possibilities, one can show that no pair of trees with more than 2 binding edges in $T_{(000)}$ or fewer than 5 in $T_{(111)}$ are uniquely constructable. 
Just as before, these pairs can again be expanded into full uniquely-constructable tree families. 
Thus, the minimum possible variance in $G^\alpha$ is achieved when $G_{(000)}^\alpha=2$ and $G_{(111)}^\alpha=5$, which yields a minimum sensitivity of $\partial_{\lambda}\langle Q\rangle=3/4$, at least at the considered distribution $\pi_{(000)}=\pi_{(111)}=1/2$.

\subsection*{Discussion}

We have identified a potential simple geometric architecture for the space of responses, providing analytic and numerical evidence.
This space is organized by a collection of optimal models that serve as the vertices of the smallest enclosing polytope.
Of note, these optimal models and their responses are determined by the topology of the state-space graph ${\mathcal G}$ through the uniquely-constructable tree families. 

This geometric perspective further suggests a procedure for numerically optimizing response.
Namely, assemble all the uniquely-constructable tree families.
Once they are known the optimal response is guaranteed to be found among them for any observable. 
While it may be computationally challenging to construct all the optimal tree families, this approach should be contrasted with the more common approach based on numerically sampling the rates $W_{e}$.
The difficulty here can be seen from this and previous work: the optimal rates require an extreme scale separation that may be hard to find numerically.
Moreover, for every new observable the procedure needs to be repeated.

Alternatively, we have shown that one need not enumerate every possible tree family when optimizing for a particular perturbation.
In our investigation of the three site model, we were able to restrict our search to only the $(000)$ and $(111)$ rooted trees with the desired properties, which reduced the number of relevant tree families from $384^{8}$ all the way down to $24^{2}$, a reduction by almost 18 orders of magnitude.
This style of process has the advantage of incurring significantly less computational cost, but lacks the versatility that being able to list all possible tree families provides.

In the case that one is able to isolate a single, physically realized tree family via sufficiently extreme rate-scale separation (such as those given in Fig. \ref{fig:hill}), the topological scaling factor for that tree family will dictate most aspects of the response curve.
The tradeoff between the response rates at different sections of the response curve we have depicted here could allow for calculable performance restrictions for biological bandpass filter mechanisms \cite{sohka2009externally,shui2024protein}.
Though, such insights would be aided by determining how to use this methodology to tune the leakage and saturation levels~\cite{ang2013tuning}.
This is exemplified in Fig.\ \ref{fig:hill} wherein $\langle Q\rangle$ is seen to vary over the entire allowable range of $\lbrack 0,1\rbrack$ despite only a single tree family being physically realized.

The most obvious next step for progressing this work would be to formally prove our conjecture that \eqref{eq:convex} and in turn verify its crucial consequences such as \eqref{eq:ineq}.
Beyond this, it would also be desirable to understand how the geometric structure of the response space, $\mathcal{R}$, changes when additional constraints or symmetries are included in the dynamics.
For example, energy is a common resource constraint in not just biophysical models, but nonequilibrium systems in general.
A proper understanding of how restrictions on energy consumption and thermodynamic forcing condenses $\mathcal{R}$ could provide tighter bounds on the system response while remaining agnostic to the individual rate kinetics.
Furthermore, when transitions emerge from related microscopic processes, such as in chemical reaction networks, the rates cannot vary independently as considered here, thereby similarly reshaping the response space.
These each represent multiple approaches to the larger problem of discerning what regions of the response space become inaccessible when restrictions are placed on the allowed dynamics of the system.

\section*{Materials and Methods}
\subsection*{Matrix Tree Theorem}
The Matrix Tree Theorem (MTT) provides a graphical representation of the steady-state solution of \eqref{eq:master} in terms of the weights of spanning trees~\cite{schnakenberg1976,Hill}:
\begin{equation}\label{eq:MTT}
\pi_j = \frac{\sum_T\omega(T_j)}{\sum_{T,l} \omega(T_l)}=\frac{\omega({\mathcal T}_j)}{\omega(\mathcal T)},
\end{equation}
where we have introduced the tree-set weights $\omega(\mathcal{T}_j)=\sum_T\omega(T_j)$, and $\omega({\mathcal T})=\sum_{l,T}\omega(T_l)$.

\subsection*{Sketch of Eq. (\ref{eq:convex_full}) derivation}
Equation (\ref{eq:MTT}) is particularly convenient for calculating response graphically.
The full details can be found in the SI.  Here we sketch the key steps.

From the MTT (\ref{eq:MTT}), the steady-state logarithmic sensitivities can be expressed as 
\begin{equation}\label{eq:MTTresponse}
\partial_\lambda \ln \pi_j=\partial_\lambda \ln \omega({\mathcal T}_j)-\partial_\lambda \ln \omega({\mathcal T}).
\end{equation} 
Thus, we need to calculate the derivative
\begin{equation}
\partial_\lambda \ln \omega({\mathcal T}_j)=\frac{1}{\omega({\mathcal T}_j)}\sum_{e,T} \frac{\partial \ln W_e}{\partial\lambda}\frac{\partial \omega(T_j)}{\partial \ln W_e}.
\end{equation}
Using the definition of $d_e$ (\ref{eq:d}) and noting that only trees $T^e_j$ that contain the edge $e$ have nonzero derivative we have 
\begin{align}\label{eq:dlwj}
\partial_\lambda \ln \omega({\mathcal T}_j)&=\frac{1}{\omega({\mathcal T}_j)} \sum_{e} d_e \sum_{T\ni e} \omega(T^e_j)
\end{align}
The last step is to exchange the  $e$ and $T$ sums and multiply by one,
\begin{align}
\partial_\lambda \ln \omega({\mathcal T}_j)=\frac{\prod_{l\neq j} \omega({\mathcal T}_l)}{\prod_{k} \omega({\mathcal T}_k)} \sum_{T} \omega(T_j) \sum_{e\in T} d_e=\sum_\sigma \frac{\omega(\sigma)}{{\mathcal Z}} G^\sigma_j,
\end{align}
recognizing ${\mathcal Z} = \prod_k \omega({\mathcal T}_k)=\sum_\sigma\omega(\sigma)$.
A similar calculation for $\partial_\lambda \ln\omega({\mathcal T})$ leads to the final result in \eqref{eq:convex_full}.

\subsection*{Uniquely-constructable steady state} By tuning the rates, we can realize a uniquely-constructable family $\alpha$.
In the SI `every physical-realizable tree family is uniquely constructable', we show this requires that each tree in $\alpha$ has the highest weight, $\omega(T_j^\alpha)\gg\omega(T_j)$ for all $T_j \neq T^\alpha_j$.
This condition simplifies the MTT expression for the steady-state distribution, since the sums in \eqref{eq:MTT} are dominated by the largest terms
\begin{equation}
\pi_j \approx \frac{\omega(T^\alpha_j)}{\sum_{l} \omega(T^\alpha_l)}.
\end{equation}
Such a steady-state changes in a rather simple manner as we vary the external parameters.  
Indeed, for constant $d_e$ as we assume in this section, varying the external parameter exponentially reweights the rates $W_e \to W_e e^{\lambda d_e}$ with a corresponding reweighting of the tree weights and the steady-state distribution as
\begin{equation}
\pi_j \to \frac{\omega(T^\alpha_j)e^{\lambda G_j^\alpha}}{\sum_{l} \omega(T^\alpha_l)e^{\lambda G_l^\alpha}},
\end{equation}
which corresponds to \eqref{eq:uniqConstDistr}.
This form is approximately valid for all $\lambda$ over which the relation $\omega(T_j^\alpha)e^{\lambda G^\alpha_j}\gg\omega(T_j)e^{\lambda G_j}$.
An immediate consequence is that the realization range, $\mathcal{L}$, may be defined as the range of $\lambda$ values wherein $\omega(\alpha)/\mathcal{Z}\approx 1$.

\subsection*{Generation of Fig. \ref{fig:convex_ver}}

The axes vectors, $\vec{b}_{1}$ and $\vec{b}_{2}$ are orthonormal basis vectors of the space satisfying $\langle \partial_{\lambda}\ln\pi\rangle=0$. Specifically, we have chosen these to take the forms

\begin{subequations}
\begin{equation}
\vec{b}_{1} = \frac{\left\lbrack\pi_{1},-\pi_{0},0\right\rbrack^{\text{T}}}{\sqrt{\pi_{0}^{2}+\pi_{1}^{2}}},
\label{bdef1}
\end{equation}
\begin{equation}
\vec{b}_{2} = \frac{\left\lbrack\pi_{0}\pi_{2},\pi_{1}\pi_{2},-\pi_{0}^{2}-\pi_{1}^{2}\right\rbrack^{\text{T}}}{\sqrt{\left(\pi_{0}^{2}+\pi_{1}^{2}\right)\left(\pi_{0}^{2}+\pi_{1}^{2}+\pi_{2}^{2}\right)}}.
\label{bdef2}
\end{equation}
\label{bdefs}
\end{subequations}

Parameter values used to generate the four plots are:

\noindent Fig.\ \ref{fig:convex_ver}A:
\begin{center}
\begin{tabular}{ c c c }
$\pi_{0} = 0.165$ & $\pi_{1} = 0.170$ & $\pi_{2} = 0.665$ \\
$d_{01} = 0.158$ & $d_{12} = 0.136$ & $d_{20} = -1.194$ \\
$d_{10} = 0.777$ & $d_{02} = 0.037$ & $d_{21} = -0.240$
\end{tabular}
\end{center}
Fig.\ \ref{fig:convex_ver}B:
\begin{center}
\begin{tabular}{ c c c }
$\pi_{0} = 0.506$ & $\pi_{1} = 0.035$ & $\pi_{2} = 0.459$ \\
$d_{01} = -0.741$ & $d_{12} = 0.874$ & $d_{20} = 0.134$ \\
$d_{10} = 1.806$ & $d_{02} = 2.091$ & $d_{21} = 0.153$
\end{tabular}
\end{center}
Fig.\ \ref{fig:convex_ver}C:
\begin{center}
\begin{tabular}{ c c c }
$\pi_{0} = 0.933$ & $\pi_{1} = 0.066$ & $\pi_{2} = 0.001$ \\
$d_{01} = -1.054$ & $d_{12} = 0.398$ & $d_{20} = -2.148$ \\
$d_{10} = 0.779$ & $d_{02} = -2.093$ & $d_{21} = 1.057$
\end{tabular}
\end{center}
Fig.\ \ref{fig:convex_ver}D:
\begin{center}
\begin{tabular}{ c c c }
$\pi_{0} = 0.437$ & $\pi_{1} = 0.179$ & $\pi_{2} = 0.384$ \\
$d_{01} = 0.030$ & $d_{12} = -0.168$ & $d_{20} = -1.387$ \\
$d_{10} = 0.558$ & $d_{02} = -0.738$ & $d_{21} = -1.205$
\end{tabular}
\end{center}

To perform the parameter sweep that generated the blue points, we begin from \eqref{eq:dlwj} and note that for the 3 state graph in particular, the ratios $\sum_{T\ni e}\omega(T_{j}^{e})/\omega(\mathcal{T}_{j})$ can always be expressed entirely in terms of $W_{20}/W_{10}$, $W_{01}/W_{21}$, and $W_{12}/W_{02}$ without any further dependence on the bare rate values themselves.
Conversely, from the MTT we know that the steady-state probabilities cannot be written in terms of only these three ratios but rather requires the bare rates in addition.
This allows us to sweep over the response space by defining the 3-dimensional space in which $\ln(W_{20}/W_{10})$, $\ln(W_{01}/W_{21})$, and $\ln(W_{12}/W_{02})$ are taken as the position coordinates while the steady-state distribution is held fixed.
We uniformly scan an octahedral region with major axes corresponding to the range $\lbrack -12,12\rbrack$ along the three coordinate axes and a grid spacing of 0.25 and calculate the response factors via \eqref{eq:MTTresponse} and (\ref{eq:dlwj}).


\section*{Acknowledgments}
This material is based upon work supported in part by the National Science Foundation under Grant No. DMR-2142466 and by the Alfred P. Sloan Foundation under Grant No. 2022-19440.

\bibliography{arxiv_v1.bib}

\onecolumngrid
\renewcommand\theequation{S.\arabic{equation}}
\setcounter{equation}{0} 
\renewcommand\thefigure{S.\arabic{figure}}
\setcounter{figure}{0}  

\section*{Supplementary Information}

\subsection*{Setup and notation}

We begin by establishing some notation and collecting relevant background results for easy reference. 

Recall that we are considering a discrete-state Markov process that admits a weighted graph representation ${\mathcal G}$.
Here, the nodes $i=1,\dots,N$ represent states and the directed edges $e$ represent possible transitions with weights $W_{e}$ given by the transition rates.
The rates are assumed to depend on an externally-controlled parameter $\lambda$ that we perturb to interrogate the dynamics.

Here we review key topological properties of ${\mathcal G}$ important for our discussion. 
These notions are illustrated with a concrete example in Fig. \ref{treedefs}.\newline

\noindent{\it Rooted trees.---}Central to our analysis are spanning trees $T_r$, which are connected subgraphs of ${\mathcal G}$ that contain every vertex, no cycles, and every edge is directed toward the root $r$.
Their weight is inherited  from ${\mathcal G}$, being given as the product of weights of every edge,
\begin{equation}
\omega\left(T_r\right) = \prod_{e\in T_r}W_{e}.
\label{wdef}
\end{equation}
This construction can naturally be extended to sets of rooted trees by summing over the collection.
Namely, the weight of the set ${\mathcal T}_r$ of all trees rooted at $r$ is $\omega({\mathcal T}_r)=\sum_{T_r\in {\mathcal T}_r} \omega(T_r)$; and similarly for the weight of the set of all rooted trees ${\mathcal T}=\bigcup_r {\mathcal T}_r$  is $\omega({\mathcal T})=\sum_r \omega({\mathcal T}_r)$.\newline

\noindent{\it Matrix Tree Theorem.---}With this notation in hand, we can now state our main theoretical tool, the Matrix Tree Theorem~\cite{schnakenberg1976,Hill}, which is a graph-theoretic expression for the steady-state distribution
\begin{equation}
\pi_{i} = \frac{\sum_{T\in\mathcal{T}_{i}}\omega\left(T_i\right)}{\sum_{r}\sum_{T\in\mathcal{T}_{r}}\omega\left(T_r \right)} = \frac{\omega\left(\mathcal{T}_{i}\right)}{\omega\left(\mathcal{T}\right)}.
\label{MTT}
\end{equation}\newline

\noindent{\it Tree families.---}While rooted trees offer a convenient representation of the steady state, we have suggested in this paper that tree families are the natural graph-theoretic construction to capture steady-state response.  
A tree family $\sigma=(T_1^\sigma,T_2^\sigma,\dots T^\sigma_N)\in {\mathcal T}_1\times{\mathcal T}_2\times\cdots{\mathcal T}_N$ is a list of $N$ rooted spanning trees, one for each node of ${\mathcal G}$.  
We define their weight to be the product of the weight of its members: $\omega(\sigma) = \prod_{T^\sigma_l\in \sigma}\omega(T^\sigma_l)$.

We say a tree family $\sigma$ is {\it multi-constructable} when by utilizing all the edges in all the trees in $\sigma$, with their multiplicty, at least one other tree family can be constructed, say $\sigma'$.  
Put another way, $\sigma'$ can formed by swapping a collection of edges among the different rooted trees in $\sigma$.
By contrast,  a  {\it uniquely-constructable} tree family is the only tree family that can be constructed from its collection of edges.
There are no ways to swap edges among the trees to construct a new tree family.
\newline

\noindent{\it Scale of sensitivity.---}To study the dynamics, we perturb the rates $W_e(\lambda)$ by changing an external parameter.  The magnitude of the response is influenced by the scale [Eq. (2) main text],
\begin{equation}\label{dScale}
d_e=\partial_\lambda \ln W_e.
\end{equation}
Building from the perturbative-scale of each rate, we can assign a scale to the tree families through the topological scale factors [Eq. (4) main text],
\begin{equation}\label{G}
G_j^\sigma=\partial_\lambda \ln \omega(T^\sigma_j)=\sum_{e\in T^\sigma_j} d_e,
\end{equation}
for all $T_j^\sigma\in \sigma$.

\subsection*{Convex representation of response: Derivation of Eq. (5)}

Equation (5) of the main text demonstrates that the steady-state response has a convex representation. 
In this section, we provide a detailed derivation of this observation

Directly differentiating the Matrix Tree Theorem, \eqref{MTT}, leads to a convenient graph-theoretic representation of the logarithmic senstivities
\begin{equation}
\partial_{\lambda}\ln\pi_{i} = \partial_{\lambda}\ln\left(\omega\left(\mathcal{T}_{i}\right)\right)-\partial_{\lambda}\ln\left(\omega\left(\mathcal{T}\right)\right) = \frac{\partial_{\lambda}\omega\left(\mathcal{T}_{i}\right)}{\omega\left(\mathcal{T}_{i}\right)}-\sum_{j}\pi_{j}\frac{\partial_{\lambda}\omega\left(\mathcal{T}_{j}\right)}{\omega\left(\mathcal{T}_{j}\right)}.
\label{dlamlogpi}
\end{equation}
From here we focus on the first term on the right-hand side, as the second is merely its steady-state average. 

First, note that $\omega(\mathcal{T}_{i})=\sum_{T_i\in {\mathcal T}_i} \omega(T_i)$ is the sum of trees rooted at state $i$, allowing the derivative to neatly distribute to each tree. 
Thus, we first differentiate the definition of $\omega(T_i)$, \eqref{wdef},
\begin{equation}
\partial_\lambda \omega(T_i)=\partial_\lambda \left(\prod_{e\in T_i} W_e\right)=\sum_{e\in T_i} \partial_\lambda W_e \prod_{l\neq e} W_l = \omega(T_i) \sum_{e\in T_i} d_e ,
\end{equation}
where in the last equality we employed the definition of $d_e$, \eqref{dScale}.
From the linearity of the derivative we can extend this across the entire set ${\mathcal T}_i$,
\begin{equation}
\partial_{\lambda}\omega\left(\mathcal{T}_{i}\right) = \sum_{T_i\in\mathcal{T}_{i}}\omega\left(T_i\right)\sum_{e\in T_i}d_{e}.
\label{dlamwT}
\end{equation}\label{treeSetDerivative}

The next step in the evaluation of \eqref{dlamlogpi} is to divide \eqref{dlamwT} by $\omega({\mathcal T}_i)$ and multiply by one, 
\begin{equation}
\frac{\partial_{\lambda}\omega\left(\mathcal{T}_{i}\right)}{\omega\left(\mathcal{T}_{i}\right)} 
=  \left(\frac{\prod_{k\ne i}\omega\left(\mathcal{T}_{k}\right)}{\prod_{k\neq i}\omega\left(\mathcal{T}_{k}\right)}\right)\frac{1}{\omega\left(\mathcal{T}_{i}\right)}\sum_{T_i\in\mathcal{T}_{i}}\omega\left(T_i\right)\sum_{e\in T_i}d_{e}
= \frac{\prod_{k\ne i}\omega\left(\mathcal{T}_{k}\right)}{\prod_{k}\omega\left(\mathcal{T}_{k}\right)}\sum_{T\in\mathcal{T}_{i}}\omega\left(T_i\right)\sum_{e\in T_i}d_{e}.
\label{dwti_exp}
\end{equation}
 First, we evaluate the denominator, $\mathcal{Z}=\prod_{k}\omega(\mathcal{T}_{k})$. Each tree set weight $\omega(\mathcal{T}_{k})$ is  a sum over all trees rooted at $k$, making $\mathcal{Z}$ a product of sums. 
 When this product is expanded out, the result is a sum of products of $N$ tree weights, comprised of one tree for every root; namely, the weights of the tree families $\omega(\sigma)$.
Crucially, since $\mathcal{T}_{k}$ contains every possible $k$-rooted tree once, our expansion of $\mathcal{Z}$ must sum over every possible tree family exactly once, leading to ${\mathcal Z}=\sum_\sigma \omega(\sigma)$.
For the numerator $\sum_{T_i\in {\mathcal T}_i} \prod_k \omega(T_k)\sum_{e\in T_i} d_e$ the analysis is nearly identical, except each term in the sum is multiplied by the topological scaling factor at that root, \eqref{G}.

Putting this all together, we find
\begin{equation}
\frac{\partial_{\lambda}\omega\left(\mathcal{T}_{i}\right)}{\omega\left(\mathcal{T}_{i}\right)} = \frac{1}{\mathcal{Z}}\sum_{\sigma}\omega\left(\sigma\right)G_{i}^{\sigma}
\label{dwti_tf}
\end{equation}
\begin{equation}
\mathcal{Z} = \prod_{k}\omega\left(\mathcal{T}_{k}\right) = \sum_{\sigma}\omega\left(\sigma\right),
\label{Zexp}
\end{equation}
which upon substitution into \eqref{dlamlogpi} produces Eq. (5) of the main text,
\begin{equation}\label{convex}
\partial_\lambda\ln\pi_i = \sum_\sigma \frac{\omega(\sigma)}{\mathcal Z}\left(G^\sigma_i-\langle G^\sigma\rangle\right).
\end{equation}
reprinted here for convenience.

The nature of \eqref{convex} as a convex representation encourages us to define its possible values within a geometric space.
Specifically, let the response vector, $\vec{r}:\{W_{e}|\pi,d_{e}\}\to\mathbb{R}^{N}$, be a mapping from the space of all possible transition rates under a given steady state distribution and set of perturbation scalings to an $N$ dimensional vector space such that $r_{i}=\partial_{\lambda}\ln (\pi_{i})$.
As we shall soon show though, the range of $\vec{r}$ does not extend to the entirety of $\mathbb{R}^{N}$.
Rather it is a finite subset of dimension at most $N-1$ which we denote as the \textit{response space}, $\mathcal{R}$.
As noted in the main text, $\vec{r}$ is restricted by the normailzation condition, $\langle\partial_{\lambda}\ln (\pi_{i})\rangle=\sum_{i}\pi_{i}r_{i}=0$, meaning that it is confined to the ``0-mean hyperplane", $\mathcal{H}=\{\vec{x}\in\mathbb{R}^{N}|\sum_{i}\pi_{i}x_{i}=0\}$.
Crucially, the tree family vertices, $v_{j}^{\sigma}=G_{j}^{\sigma}-\langle G^{\sigma}\rangle$, also satisfy this condition and are thus similarly confined to $\mathcal{H}$.
As such, when the $N-1$ dimensional convex hull of $\{v_{j}^{\sigma}\}$ is taken within $\mathcal{H}$, $\vec{r}$ must be further confined to be within this hull due to the convex nature of \eqref{convex}, thus limiting its possible locations (and by definition the entirety of $\mathcal{R}$) in $\mathcal{H}$ to such a finite subset.
In the main text, we conjecture that $\mathcal{R}$ can be further restricted to to convex hull of only the uniquely constructable tree family vertices due to the equivalency of unique constructability and physical realizability, as we now show.

\subsection*{Every physical-realizable tree family is uniquely constructable}
As discussed in the main text, the fact that \eqref{convex} is a convex combination of all the topological scaling factors does not guarantee that the response can achieve the values $\partial_\lambda\ln \pi_j = v^\sigma_j$ for some collection of rates: not all tree families are {\it physically realizable}.
In this section, we prove a graphical characterization of all physically-realizable tree families by demonstrating that they must be uniquely constructable:
\begin{theorem}
A tree family is physically realizable if and only if it is uniquely constructable.
\label{constructtheorem}
\end{theorem}

The proof of Theorem \ref{constructtheorem} will require introducing and proving a number of concepts.
The first is to precisely define physical realizablility, recognizing that it is a property that can only be arbitrarily-well approximated for any particular rate matrix:
\begin{definition}
A tree family $\sigma$ is physically realizable if for every $\delta>0$ there exists a set of transition rates, $\{W_{e}\}$, such that $\omega(\sigma)/\mathcal{Z}>1-\delta$.
\label{physreal}
\end{definition}
 As $\delta$  gets smaller, $\omega(\sigma)$ must take up more of the total tree family weight $\mathcal{Z}$, meaning that it must become the most highly-weighted tree family, dominating the other tree families.
This suggests the following definition:
\begin{definition}
For a given set of rates $\{W_e\}$, a tree family $\sigma$ is said to be dominant if the weight of every $T_i^\sigma\in \sigma$ is the largest, that is $\omega(T_{i}^\sigma)>\omega(T_{i}),\ \forall\ T_{i}\in\mathcal{T}_{i}$ and $T_{i}\neq T_{i}^\sigma$.
\label{dominant}
\end{definition}
\noindent Thus a necessary condition for a tree family to be physically realizable is that it must be possible to be a dominant tree family, but as we now show all dominant tree families are uniquely constructable.
This is to the first part of Theorem \ref{constructtheorem}:
\begin{lemma}
Every dominant tree family is uniquely constructable.
\label{edgeswap}
\end{lemma}
\begin{proof}
Let $\sigma$ be a dominant tree family. Since the weight of a tree is the product of the weight of its edges, the weight of the tree family, $\omega(\sigma)=\prod_{T_i\in\sigma}\omega(T_i)$, is independent of how the trees are constructed from the available edges. Thus, if there was another tree family, $\sigma'\neq\sigma$, that uses the exact same edges with the same multiplicities it would have the same weight, $\omega(\sigma')=\omega(\sigma)$. This violates the assumption that $\sigma$ has the highest weight, implying that there cannot be any construction of trees other than $\sigma$ from the available edges; or equivalently, $\sigma$ is uniquely constructable.
\end{proof}

Next, we prove the converse, which we will accomplish by explicitly constructing a collection of rates for which a tree family $\sigma$ is physically realizable.
At a minimum, this means choosing rates so that for every $T_i^\sigma\in \sigma$, we have $\omega(T_i^\sigma)>\omega(T_i)$.
This is easy to arrange for tree families that contains edges not in $\sigma$.
We simply set the associated $W_e=0$, forcing all those tree families to have zero weight while maintaining $\omega(\sigma)>0$. 
The challenge arises when we consider those tree families $\sigma'$ that are built entirely from edges in $\sigma$ but with different multiplicities.
The weights of these tree families must be nonzero, $\omega(\sigma')>0$, but we need to arrange it so that $\omega(\sigma')<\omega(\sigma)$ as well.
Here, the idea is to compare the weight of each $T_i^\sigma\in \sigma$ to the weight of all other trees that can be built from the edges in $\sigma$.
This leads to constraints that the weights must satisfy:
\begin{definition}
$n^{th}$-order local-rate constraint: Consider a tree in a dominant tree family, $T_i^\sigma \in \sigma$ and a distinct tree rooted at the same node $T'_i$, but composed only of edges present in $\sigma$.
The dominance of $\sigma$ requires $\omega(T_{i}^\sigma)>\omega(T_{i}')$. 
Removing all common-edge weights reduces this condition to $\prod_{e\in\mathcal{E}^\sigma}W_{e}>\prod_{e\in\mathcal{E}'}W_{e}$, where $\mathcal{E}^\sigma$ is the set of edges that are in $T_{i}^\sigma$ but not $T_{i}'$ and $\mathcal{E}'$ is the set of edges in $T_{i}'$ but not $T_{i}^\sigma$. If there are $n$ such differing edges, we call this inequality an $n^{th}$-order local-rate constraint.
\label{swapcon}
\end{definition}
\noindent Formulated another way, the edges $\mathcal{E}^\sigma \in T_i^\sigma$ can be swapped out and replaced by $\mathcal{E}'$ to form a new tree $T_i'=(T_i^\sigma\setminus\mathcal{E}^\sigma)\cup\mathcal{E}'$.
The newly formed tree must then have a lower weight.
For example, by letting $W_{ab}$ represent a particular $\{b\to a\}$ transition rate, every $2^{\text{nd}}$-order local-rate constraint will come from examining some $T_i^\sigma$, and have the form 
\begin{equation}
W_{l j}W_{mk}>W_{nj}W_{pk},
\end{equation}
comparing rates exiting 2 distinct nodes, here $j$ and $k$.
The implication is that it is possible to remove the upper-bound edges $\{\{j\to l\},\{k\to m\} \}$ from $T_i^\sigma$ and replace them with the lower bound edges $\{\{j\to n\},\{k\to p\}\}$ to form a valid tree still rooted at $i$, but not in $\sigma$.

For a tree family to be dominant every local-rate constraint must be satisfied at every order $n$ and at every node $i$.  
It is not hard to imagine that we can choose our rates to satisfy these constraints locally at a particular node.
But they still must be consistent globally across all the nodes in the network.
That this is possible for a uniquely-constructable tree is the content of the following:
\begin{lemma}
Every uniquely-constructable tree family is physically realizable.
\label{uniqueConstruct}
\end{lemma}
\begin{proof}
Consider a uniqeuly-constructable family $\sigma$. For every edge not in $\sigma$, set $W_e=0$, making every tree family $\sigma'$ with an edge not in $\sigma$ have zero weight $\omega(\sigma)>\omega(\sigma')=0$.
For $\sigma$ to be dominant, the remaining $\{W_e\}$ must satisfy the finite system of inequalities generated by all the local-rate constraints at every order $n$ and node $i$.
The first thing we need to check is that this system of inequalities is feasible, that is admits a solution.  
All local-rate constraints that come from the same tree $T_i^\sigma\in\sigma$ are internally consistent, as the the upper-bound and lower-bound edges are disjoint: the upper-bound edges must be in $T_i^\sigma$, while the lower-bound edges necessarily are not.
An inconsistency could arise between the local-rate constraints generated by different trees in $\sigma$ that have distinct roots.
This would occur if there were a collection of local-rate constraints each generated by separate trees in $\sigma$ where every upper-bound edge appears in another inequality as a lower-bound edge.
However, each inequality represents a potential exchange of edges where the upper-bound edges can be swapped out for the lower-bound edges.
If the edges appear both as lower and upper bounds, this would imply that each edge appears in more than one tree and that those edges could be swapped between those trees allowing the construction of another valid tree family.
As this would violate the unique-constructability of $\sigma$, we can conclude that the system of local-rate constraints are indeed feasible.

Having established that there exists some collection of rates $\{W_e\}$ for which $\sigma$ is dominant, all that is left is to exhibit a particular solution that agrees with Definition \ref{physreal} for physical realizability.
To this end, let us index all the local-rate constraints generated when comparing $\sigma$ to every nondominant tree family, $\beta\in B(\sigma)$, and associate each with a constraint ratio $\epsilon_{\beta}=\prod_{e\in E'}W_{e}/\prod_{e\in E^{\sigma}}W_{e}<1$.
We further note that for any feasible set of rates $\{W_e\}$, we can generate another feasible collection via a power-law rescaling $W_{e}\to W_{e}^{\gamma}$ for some $\gamma>0$, as positive powers are monotonic with respect to the inequalities in the local-rate constraints.
Such a scaling also transforms the constraint ratios as $\epsilon_{\beta}\to\epsilon_{\beta}^{\gamma}$.

%

From here we compare the weights of the nondominant tree families to the dominant tree family by considering how each nondominant tree family $\sigma'$ can be transformed into the dominant one, $\sigma$.
One way to do so is to consider each tree in $\sigma'$ individually and, using the same notation as in Definition \ref{swapcon}, exchange all those edges in $E'$ for $E^{\sigma}$ that appeared in the corresponding local-rate constraint.
This transforms $\omega(\sigma')$ by dividing away all rates that come from edges in $E'$ and multiplying by all rates that come from edges in $E^\sigma$, which is equivalent to dividing by the constraint ratio, $\epsilon_{\beta}$.
Repeating this for each tree in the family then yields the relationship $\omega(\sigma')=\omega(\sigma)\prod_{\beta\in B(\sigma'\to\sigma)}\epsilon_{\beta}$ where $B(\sigma'\to \sigma)\subseteq B(\sigma)$ is the set of local-rate constraints that derive from comparing $\sigma'$ to $\sigma$.
Since each $\epsilon_{\beta}$ is necessarily smaller than 1, this immediately implies $\omega(\sigma')\le \epsilon_{M}\omega(\sigma)$, where $\epsilon_{M}=\max_{\beta\in B(\sigma)}\epsilon_{\beta}$.
We next rescale all the rates $W_e \to W_e^\gamma$, so that for the new feasible collection of rates the weights of the tree families verify
$\omega(\sigma')\le \epsilon_{M}^{\gamma} \omega(\sigma)$.
Thus, what we have constructed is a collection of rates for which  the $S=||\{\sigma\}||-1$ nondominant tree families have weights that are all smaller than that of the dominant tree family by a factor of $\epsilon_{M}^{\gamma}$, which can be made arbitrarily small by increasing $\gamma$.
Together, these observations allow us to bound the contribution of $\sigma$ to the sum in \eqref{convex} as
\begin{equation}
\frac{\omega(\sigma)}{{\mathcal Z}} = \frac{\omega(\sigma)}{\omega(\sigma)+\sum_{\sigma'\neq\sigma}\omega(\sigma')}\ge \frac{1}{1+S\epsilon_{M}^{\gamma}}> 1- S\epsilon_{M}^{\gamma},
\end{equation}
So, we can satisfy the definition of physical realizability for any $\delta$ by choosing $\gamma =\log_{\epsilon_{M}}(\delta/S)$.
\end{proof}

\begin{proof}[Proof of Theorem \ref{constructtheorem}] Follows from  Lemma \ref{edgeswap} and Lemma \ref{uniqueConstruct}.
\end{proof}

\subsection*{Convex parameters for the 3-state graph}
In Eq. (6) of the main text we posit that the logarithmic response can be written as convex sum of the topological scale factors for only the physically-realizable tree families: 
\begin{equation}\label{convexUnique}
\partial_\lambda \ln \pi_j = \sum_\alpha \theta^\alpha (G^\alpha_j - \langle G^\alpha\rangle),
\end{equation}
which should be contrasted  with \eqref{convex} where the sum is over every tree family.
As evidence of the veracity of this conjecture, we construct a valid collection of $\theta^\alpha$ for the $N=3$ triangle graph depicted in Fig. \ref{treedefs}.
We accomplish this goal by rearranging the sum in \eqref{convex} by re-expressing the topological scaling factors for the multi-constructable families in terms of the uniquely-constructable ones.

We start by listing all 27 possible tree families, 11 of which are uniquely constructable (Fig. \ref{treefams_uc}). 
The remaining 16 multi-constructable families can be split into 7 different groups that we call \textit{equivalency classes}, since they are formed from the same collection of edges and thus all have the same weight.
There are 6 total equivalency classes with 2 tree families each (Fig. \ref{treefams_mc2}) and a single equivalency class with 4 tree families (Fig. \ref{treefams_mc4}).

From these graphs, the various topological scaling factors can be deduced. For instance, in family 0 (top row, Fig. \ref{treefams_uc}), the 0-rooted tree contains the $\{1\to 0\}$ and $\{2\to 1\}$ edges, giving $G_{0}^{(0)}=d_{01}+d_{12}$. We can also note that families 1, 2, and 5 also use this same 0-rooted tree and thus have the same topological scaling factor at state 0 ($G_{0}^{(0)}=G_{0}^{(1)}=G_{0}^{(2)}=G_{0}^{(5)}$). We can use this method of generating the topological scaling factors to more directly compare the uniquely and multi-constructable tree families. 
As an example, take the equivalency class $\{11,12\}$ (first two rows, Fig. \ref{treefams_mc2}). 
Notice that by exchanging the 0-rooted trees between the two families that comprise this equivalency class we reproduce the two uniquely-constructable families 3 and 6.
We denote this congruence as $\{11,12\}\cong\{3,6\}$. 
This further implies the sums of topological scaling factors are equal,  $G_{i}^{(11)}+G_{i}^{(12)}=G_{i}^{(3)}+G_{i}^{(6)}$ for any $i$. 
Notably, such congruencies with uniquely constructable families exist for every size 2 equivalency class;
\begin{align}
&\{11,12\}\cong\{3,6\}, \quad \{13,14\}\cong\{1,4\}, \quad \{15,16\}\cong\{2,5\}, \nonumber\\
&\{17,18\}\cong\{5,8\}, \quad \{19,20\}\cong\{6,9\}, \quad \{21,22\}\cong\{4,7\}.
\label{mc2congs}
\end{align}
Since tree families within the same equivalency class share a weight, we can simplify their contribution to sum in \eqref{dwti_tf}. Continuing with the example of the $\{11,12\}$ equivalency class with $w^{(11)}$ and $w^{(12)}$, this takes the form
\begin{equation}
\frac{1}{\mathcal{Z}}\left(w^{\left(11\right)}G_{i}^{\left(11\right)}+w^{\left(12\right)}G_{i}^{\left(12\right)}\right) = \frac{w^{\left(11\right)}}{\mathcal{Z}}\left(G_{i}^{\left(11\right)}+G_{i}^{\left(12\right)}\right) = \frac{w^{\left(11\right)}}{\mathcal{Z}}\left(G_{i}^{\left(3\right)}+G_{i}^{\left(6\right)}\right).
\label{1112to36}
\end{equation}
This observation implies that the topological scaling factors for the $\{11,12\}$ equivalency class can be removed from the sum.
Continuing in this way, the topological scaling factors for every of size 2 equivalency classes can be entirely removed from the sum in \eqref{dwti_tf}.

Unfortunately, the size 4 equivalency class, $\{23,24,25,26\}$ (Fig. \ref{treefams_mc4}), is notably more complicated. There is no combination of 4 uniquely constructable families that is congruent to this equivalency class. However, we can combine it with families 0 and 10 to produce the relation $\{0,10,23,24,25,26\}\cong\{1,2,3,7,8,9\}$, which in turn provides us with

\begin{align}
&\frac{1}{\mathcal{Z}}\left(w^{\left(23\right)}G_{i}^{\left(23\right)}+w^{\left(24\right)}G_{i}^{\left(24\right)}+w^{\left(25\right)}G_{i}^{\left(25\right)}+w^{\left(26\right)}G_{i}^{\left(26\right)}\right) = \nonumber\\
&\frac{w^{\left(23\right)}}{\mathcal{Z}}\left(G_{i}^{\left(1\right)}+G_{i}^{\left(2\right)}+G_{i}^{\left(3\right)}+G_{i}^{\left(7\right)}+G_{i}^{\left(8\right)}+G_{i}^{\left(9\right)}-G_{i}^{\left(0\right)}-G_{i}^{\left(10\right)}\right)
\label{alltrunksGeq}
\end{align}

\noindent Combining all of these results together transforms \eqref{dwti_tf} into

\begin{align}
\frac{\partial_{\lambda}\omega\left(\mathcal{T}_{i}\right)}{\omega\left(\mathcal{T}_{i}\right)} = \frac{1}{\mathcal{Z}}&\left\lbrack\left(w^{\left(0\right)}-w^{\left(23\right)}\right)G_{i}^{\left(0\right)}+\left(w^{\left(1\right)}+w^{\left(13\right)}+w^{\left(23\right)}\right)G_{i}^{\left(1\right)}+\left(w^{\left(2\right)}+w^{\left(15\right)}+w^{\left(23\right)}\right)G_{i}^{\left(2\right)}+\right. \nonumber\\
&\left(w^{\left(3\right)}+w^{\left(11\right)}+w^{\left(23\right)}\right)G_{i}^{\left(3\right)}+\left(w^{\left(4\right)}+w^{\left(13\right)}+w^{\left(21\right)}\right)G_{i}^{\left(4\right)}+\left(w^{\left(5\right)}+w^{\left(15\right)}+w^{\left(17\right)}\right)G_{i}^{\left(5\right)}+ \nonumber\\
&\left(w^{\left(6\right)}+w^{\left(11\right)}+w^{\left(19\right)}\right)G_{i}^{\left(6\right)}+\left(w^{\left(7\right)}+w^{\left(21\right)}+w^{\left(23\right)}\right)G_{i}^{\left(7\right)}+\left(w^{\left(8\right)}+w^{\left(17\right)}+w^{\left(23\right)}\right)G_{i}^{\left(8\right)}+ \nonumber\\
&\left.\left(w^{\left(9\right)}+w^{\left(19\right)}+w^{\left(23\right)}\right)G_{i}^{\left(9\right)}+\left(w^{\left(10\right)}-w^{\left(23\right)}\right)G_{i}^{\left(10\right)}\right\rbrack.
\label{dwti_basered}
\end{align}

Equation (\ref{dwti_basered}) is almost, but not quite entirely in the form dictated by \eqref{convexUnique} [Eq. 6 main text]. 
The parameters associated to tree families 0 and 10 are not guaranteed to be nonnegative. 
This can be rectified by noting a congruency within the uniquely constructable families themselves: $\{0,0,7,8,9\}\cong\{1,2,3,10,10\}$ (note the double presence of 0 and 10 on either side of this congruency, meaning two copies of each are needed). The existence of such a congruency implies
\begin{equation}
2G_{i}^{\left(0\right)}+G_{i}^{\left(7\right)}+G_{i}^{\left(8\right)}+G_{i}^{\left(9\right)}-G_{i}^{\left(1\right)}-G_{i}^{\left(2\right)}-G_{i}^{\left(3\right)}-2G_{i}^{\left(10\right)} = 0.
\label{G0comb}
\end{equation}

\noindent Since the right-hand side of \eqref{G0comb} is 0, we can freely add any multiple of its left-hand side to \eqref{dwti_basered}. Letting $\eta$ represent such a multiplicative factor finally yields

\begin{equation}
\frac{\partial_{\lambda}\omega\left(\mathcal{T}_{i}\right)}{\omega\left(\mathcal{T}_{i}\right)} = \sum_{\alpha=0}^{10}\theta^{\left(\alpha\right)}G_{i}^{\left(\alpha\right)},
\label{dwti_fullred}
\end{equation}

\noindent where

\begin{align}
&\theta^{\left(0\right)} = \frac{1}{\mathcal{Z}}\left(w^{\left(0\right)}-w^{\left(23\right)}+2\eta\right), &&\theta^{\left(1\right)} = \frac{1}{\mathcal{Z}}\left(w^{\left(1\right)}+w^{\left(13\right)}+w^{\left(23\right)}-\eta\right), \nonumber\\
&\theta^{\left(2\right)} = \frac{1}{\mathcal{Z}}\left(w^{\left(2\right)}+w^{\left(15\right)}+w^{\left(23\right)}-\eta\right), &&\theta^{\left(3\right)} = \frac{1}{\mathcal{Z}}\left(w^{\left(3\right)}+w^{\left(11\right)}+w^{\left(23\right)}-\eta\right), \nonumber\\
&\theta^{\left(4\right)} = \frac{1}{\mathcal{Z}}\left(w^{\left(4\right)}+w^{\left(13\right)}+w^{\left(21\right)}\right), &&\theta^{\left(5\right)} = \frac{1}{\mathcal{Z}}\left(w^{\left(5\right)}+w^{\left(15\right)}+w^{\left(17\right)}\right), \nonumber\\
&\theta^{\left(6\right)} = \frac{1}{\mathcal{Z}}\left(w^{\left(6\right)}+w^{\left(11\right)}+w^{\left(19\right)}\right), &&\theta^{\left(7\right)} = \frac{1}{\mathcal{Z}}\left(w^{\left(7\right)}+w^{\left(21\right)}+w^{\left(23\right)}+\eta\right), \nonumber\\
&\theta^{\left(8\right)} = \frac{1}{\mathcal{Z}}\left(w^{\left(8\right)}+w^{\left(17\right)}+w^{\left(23\right)}+\eta\right), &&\theta^{\left(9\right)} = \frac{1}{\mathcal{Z}}\left(w^{\left(9\right)}+w^{\left(19\right)}+w^{\left(23\right)}+\eta\right), \nonumber\\
&\theta^{\left(10\right)} = \frac{1}{\mathcal{Z}}\left(w^{\left(10\right)}-w^{\left(23\right)}-2\eta\right).
\label{thetaforms}
\end{align}

The formulas for each $\theta^{(\alpha)}$ given in \eqref{thetaforms} satisfy the normalization condition $\sum_{\alpha}\theta^{(\alpha)}=1$ from their construction. 
To ensure they are each $\theta^\alpha \ge 0$ requires $\eta$ to simultaneously satisfy four lower bounds and four upper bounds:
\begin{equation}
\begin{rcases}
\frac{1}{2}\left(w^{\left(23\right)}-w^{\left(0\right)}\right) \\
-\left(w^{\left(7\right)}+w^{\left(21\right)}+w^{\left(23\right)}\right) \\
-\left(w^{\left(8\right)}+w^{\left(17\right)}+w^{\left(23\right)}\right) \\
-\left(w^{\left(9\right)}+w^{\left(19\right)}+w^{\left(23\right)}\right)
\end{rcases} \le\eta\le \begin{cases}
\frac{1}{2}\left(w^{\left(10\right)}-w^{\left(23\right)}\right) \\
w^{\left(1\right)}+w^{\left(13\right)}+w^{\left(23\right)} \\
w^{\left(2\right)}+w^{\left(15\right)}+w^{\left(23\right)} \\
w^{\left(3\right)}+w^{\left(11\right)}+w^{\left(23\right)}
\end{cases}.
\label{etacons}
\end{equation}

All that is left to check is that this system of inequalities is feasible; that it admits a viable solution for $\eta$.
The bottom 3 bounds on either side of \eqref{etacons} are trivially consistent with each other, 
since each term on the left-hand side is negative while each term on the right-hand side is positive.
The upper 2 conditions complicate matters by being capable of switching between positive and negative. 
Note, however, that $\frac{1}{2}(w^{\left(23\right)}-w^{\left(0\right)})\le w^{(23)}$ and that each of the bottom 3 conditions on the right-hand side have a $w^{(23)}$.
Thus, the lower bound of $\frac{1}{2}(w^{\left(23\right)}-w^{\left(0\right)})$ is necessarily smaller than each of those 3 upper bounds. 
A similar analysis shows that the upper bound of $\frac{1}{2}(w^{\left(10\right)}-w^{\left(23\right)})$ is necessarily larger than each of the bottom 3 lower bounds.
The only point of consistency left to verify is that the upper 2 conditions are consistent with each other. To analyze these values, it is useful to note that $w^{(23)}$ is simply the geometric average of $w^{(0)}$ and $w^{(10)}$;
\begin{equation}
w^{\left(23\right)} = \sqrt{w^{\left(0\right)}w^{\left(10\right)}}.
\label{wgeomrel}
\end{equation}
With this, we can take the difference between $\frac{1}{2}(w^{\left(23\right)}-w^{\left(0\right)})$ and $\frac{1}{2}(w^{\left(10\right)}-w^{\left(23\right)})$ to show
\begin{equation}
\frac{1}{2}\left(w^{\left(10\right)}-w^{\left(23\right)}\right)-\frac{1}{2}\left(w^{\left(23\right)}-w^{\left(0\right)}\right) = \frac{1}{2}\left(w^{\left(0\right)}+w^{\left(10\right)}\right)-\sqrt{w^{\left(0\right)}w^{\left(10\right)}} = \frac{1}{2}\left(\sqrt{w^{\left(0\right)}}-\sqrt{w^{\left(10\right)}}\right)^{2} \ge 0.
\label{condiff}
\end{equation}
Equation (\ref{condiff}) proves that the upper bound of $\frac{1}{2}(w^{\left(10\right)}-w^{\left(23\right)})$ is always larger than the lower bound of $\frac{1}{2}(w^{\left(23\right)}-w^{\left(0\right)})$, meaning these conditions on $\eta$ are consistent with each other.
Furthermore, this proves that all conditions given in \eqref{etacons} can always be simultaneously satisfied, thus enabling an appropriately chosen value of $\eta$ to produce nonnegative values for each $\theta^{(\alpha)}$ in \eqref{thetaforms}.
This allows us to insert \eqref{dwti_fullred} into \eqref{dlamlogpi} to fully validate Eq. 6 of the main text for the 3 state graph.

\subsection*{Optimal uniquely-constructable tree families in the 3 site model}
In this section we construct all uniquely-constructable families for the three site binding model discussed in the main text.
Recall, the binding rates are varied uniformly and the unbinding edges are held fixed.
This is depicted Fig. \ref{unbinds} where the states are labelled according to their binary representation with $(000)$ being the fully unbound state and $(111)$ being the fully bound state while the red edges correspond to binding transitions ($d=1$) and purple edges unbinding transitions ($d=0$).

As identified in the main text, maximizing sensitivity requires a uniquely-constructable tree family where the tree rooted at $(000)$ contains only unbinding edges while the tree rooted at $(111)$ contains only binding edges. 
Beginning with the $(000)$-rooted tree, we can deduce all possible unbinding-only trees and separate them into 4 types, depicted in Fig. \ref{unbinds}.
Observe that there is only one possible unbinding edge from all the single-bound states  $\{(001), (010), (100)\}$, so any optimal $(000)$-rooted tree has those edges.
For the double-bound states $\{(011), (101), (110)\}$, there are two unbinding edges leaving, so we must choose one.
This leads to two distinct possibilities.
(i) Each double-bound state leads into a distinct single-bound state, such as shown in the Type I trees, or (ii) two of the double-bound states could lead into the same single bound state, such as in the Types II-IV trees. 

In the case of the Type I trees in option (i), once the edges leaving the double-bound states are chosen, the fully bound state $(111)$ can leave via any of the three edges connected to it as they are all unbinding edges. 
The specific tree depicted as Type I in Fig. \ref{unbinds} is one example of this, but there are in total 6 different trees that all fit this description. 
The remaining 5 can be obtained by simply permuting the ordering of the binding sites, which is geometrically equivalent to rotations and/or reflections of the state space cube that keep the $(000)$ and $(111)$ states fixed.

For option (ii), where two of the double-bound states lead into the same single bound state, the placement of the edge leaving the remaining double-bound state and the fully bound-state dictate the type of tree. 
Type II trees have the fully bound state leading into the isolated double-bound state. Once again, the tree specifically depicted in Fig.~\ref{unbinds} is one of 6 such possibilities, with the remaining ones obtainable through permutation. 
Type III and IV trees have the fully-bound state leading into one of the double-bound states that share a single bound target state with the difference between the two being whether the direction of the edge leaving the fully-bound state is perpendicular or parallel to that edge leaving the isolated double-bound state. 
Yet again, each of these representative trees belongs to a set of 6 that are all permutations of each other. 
In all, there are 24 possible trees that contain only unbinding edges (4 types with 6 trees per type).

At the same time, we need to find a corresponding $(111)$-rooted tree that both contains only binding edges and is uniquely-constructable with the $(000)$-rooted tree. 
This process is depicted in Fig. \ref{type1}. 
Given the depicted $(000)$-rooted tree, we begin constructing the $(111)$-rooted tree by incorporating the three binding edges leaving the double bound states. 
These must exist as they are the only binding edges that leave those states. 
We then focus on the binding edges that leave $(010)$ and $(100)$ states. We note that the $(010)$  state must leave through the $\{(100)\to(101)\}$ edge since if the $\{(100)\to (110)\}$ edge was used then the edge leaving the $(101)$ state could be exchanged between the two trees, thus making them not uniquely constructable. 
Similar logic dictates that the $(010)$ state must leave via the $\{(010)\to(110)\}$ edge so as to prevent the edge leaving the $(110)$ state from being exchangeable. 
From here we immediately encounter a problem when trying to deduce which edge should exit the $(000)$ state. 
Note that we can simultaneously exchange the edges leaving the $(010)$ and $(110)$ states between the two trees to produce a viable $(000)$-rooted tree and a potentially viable partial $(111)$-rooted tree. 
In order to stop this reconstruction from being allowed, the $(000)$ state must leave via the $\{(000)\to (010)\}$ edge. 
However, if we incorporate this edge, then we can simultaneously exchange the edges leaving the $(100)$ and $(101)$ states instead and produce two viable trees. 
This incompatibility means that $(000)$-rooted trees of Type I have no binding-only $(111)$-rooted trees with which they are uniquely constructable.

Performing a similar analysis on the trees of Type II yields a similar result almost immediately (Fig. \ref{type2}). 
We again begin constructing the $(111)$-rooted tree by adding the binding edges at each double bound state. 
This immediately leads to a contradiction when attempting to discern which edge exits the $(100)$ state. 
If it leaves via the $(100)\to(101)\}$ edge, then the edges leaving the $(110)$ state may be exchanged between the two trees. However, if it leaves via the $\{(100) \to (110)\}$ edge to prevent this, then the edges leaving the $(101)$ state may be exchanged. 
In either case, uniquely constructability is not achieved, meaning the $(000)$-rooted trees of Type II also have no binding-only $(111)$-rooted trees with which they are uniquely constructable.

In analyzing the Type III trees we encounter our first success (Fig. \ref{type3}).
Once the necessary edges at each double-bound states are added, we can consider all remaining possibilities. 
We first note that regardless of which edge exits the $(010)$ state it will be exchangable with an edge from the $(000)$-rooted tree unless the $(000)$ state leaves via the $\{(000) \to (010)\}$ edge. Additionally, the $(001)$ state must leave via the $(001)\to (011)$ edge to prevent an exchange at the $(011)$, while the $(100)$ state must leave via the $\{(100\to(101)\}$ edge to prevent an exchange at the $(101)$ state. 
This leaves only the edge leaving the $(010)$ state to be determined. 
To do so, we note that if the $(010)$ state were to leave via the $\{(010)\to(110)\}$ edge, then the edges leaving the $(001)$ and $(011)$ state could be exchanged. 
Utilizing the $\{(010\to(011)\}$ edge instead prevents this from happening. 
The resultant $(111)$-rooted tree is then uniquely constructable with the initial $(000)$-rooted tree. 
The remaining Type III $(000)$-rooted trees will be uniquely constructable with $(111)$-rooted trees via permutations of the tree depicted in Fig. \ref{type3}. 
These will ultimately belong to the twisted hysteresis motif discussed in the main text.

Finally, we analyze the Type IV trees and find yet another success (Fig. \ref{type4}). 
We again add the necessary edges at each double-bound state and consider all remaining possibilities. 
In much the same manner as in the case of the Type III trees, we find that the $(000)$ state must leave via the $\{(000)\to(001)\}$ edge to prevent exchangeability at the $(001)$ state, the $(010)$ state must leave via the $\{(010)\to(011)\}$ edge to prevent exchangeability at the $(011)$ state, and the $(100)$ state must leave via the $\{(100)\to(101)\}$ edge to prevent exchangeability at the $(101)$ state. With those edges in place, the edge leaving the $(001)$ state must be the $\{(001)\to(011)\}$ edge so as to prevent the possibility of a simultaneous exchange at the $(010)$ and $(011)$ states. With this, we have again created a $(111)$-rooted tree that is uniquely constructable with the initial $(000)$-rooted tree (with similar trees for other Type IV trees obtained through permutations). These follow the nested hysteresis motif.

We have thus shown that of the original 24 $(000)$-rooted trees with all unbinding edges, only 12 have corresponding uniquely constructable binding-only $(111)$-rooted trees. 
The trees rooted at the remaining states can be built using a similar analysis style to show that each of these pairs leads to a full uniquely-constructable tree family of the nested or twisted hysteresis form (with six of each all obtainable from each other through permutations). These are all depicted in Figs. \ref{3bsvert11}-\ref{3bsvert10}.

\begin{figure}[h]
\centering
\includegraphics[width=.8\textwidth]{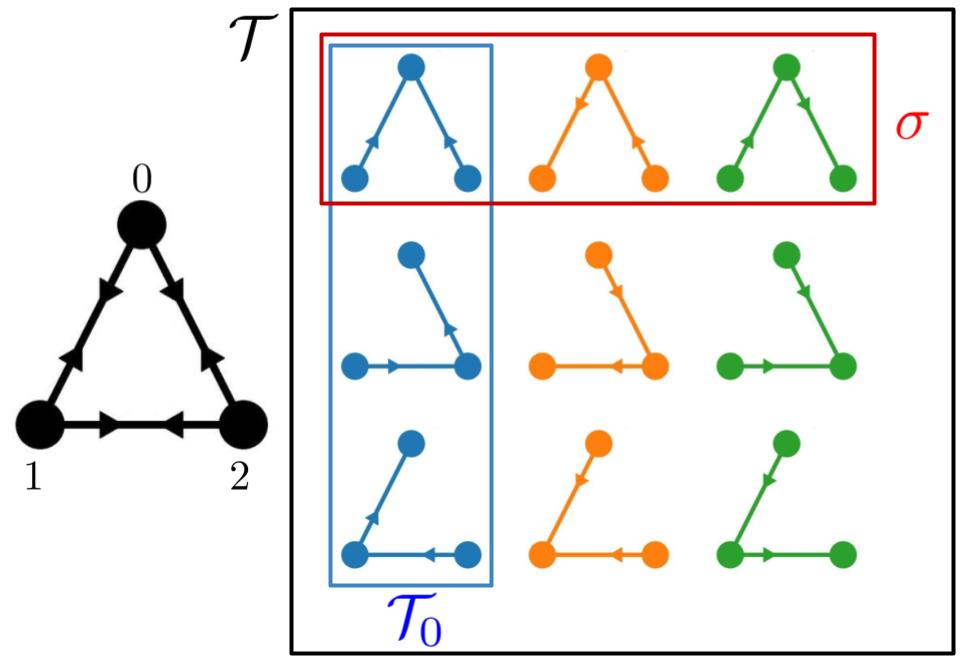}
\caption{All 9 rooted spanning trees for the three state system. The set of all such trees forms $\mathcal{T}$, while $\mathcal{T}_{0}$ is comprised of only those rooted at state $0$ (the top state). An example tree family, $\sigma$, is highlighted; one of 27 possible such tree families.}
\label{treedefs}
\end{figure}

\begin{figure}
\centering
\includegraphics[width=.85\linewidth]{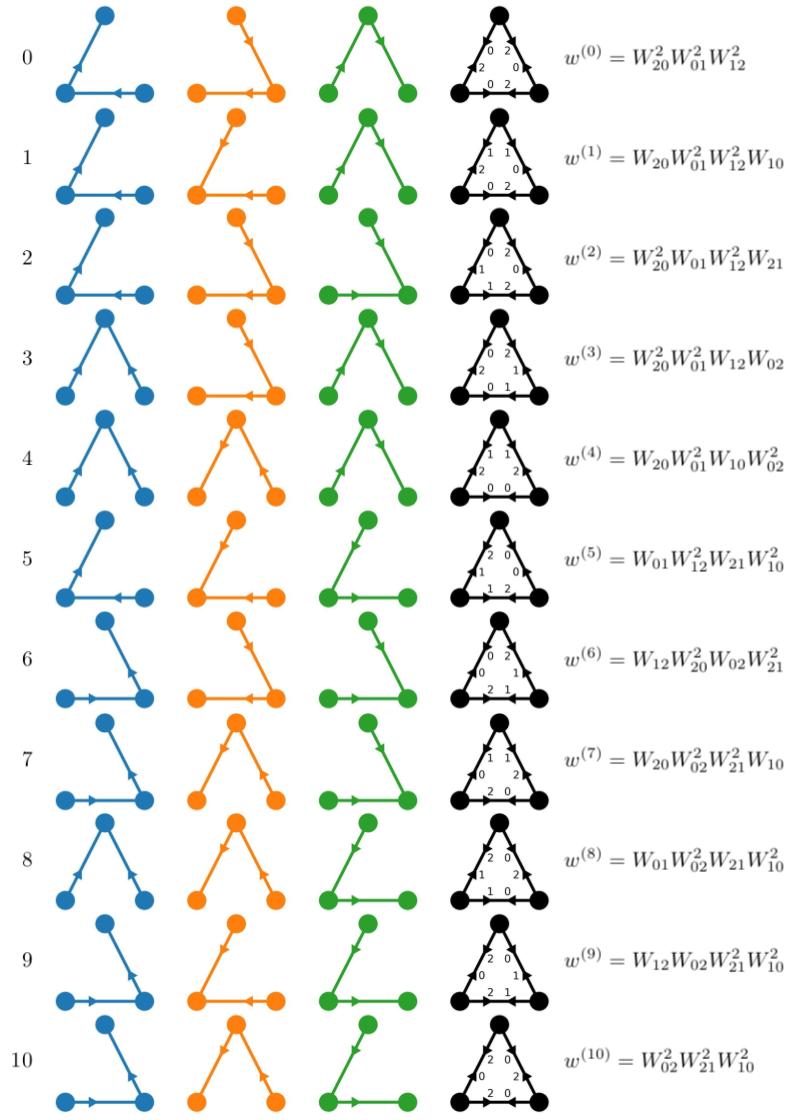}
\caption{The 11 uniquely constructable tree families. Each is given an identification number (left column) and presented such that the blue trees are rooted at the 0 state (top), the orange trees are rooted at the 1 state (left), and the green trees are rooted at the 2 state (right). The edge mulitiplicity graphs (black) then show the number of times each directed edge is used throughout the whole family, and the weight of each family is presented (right column).}
\label{treefams_uc}
\end{figure}

\begin{figure}
\centering
\includegraphics[width=.85\linewidth]{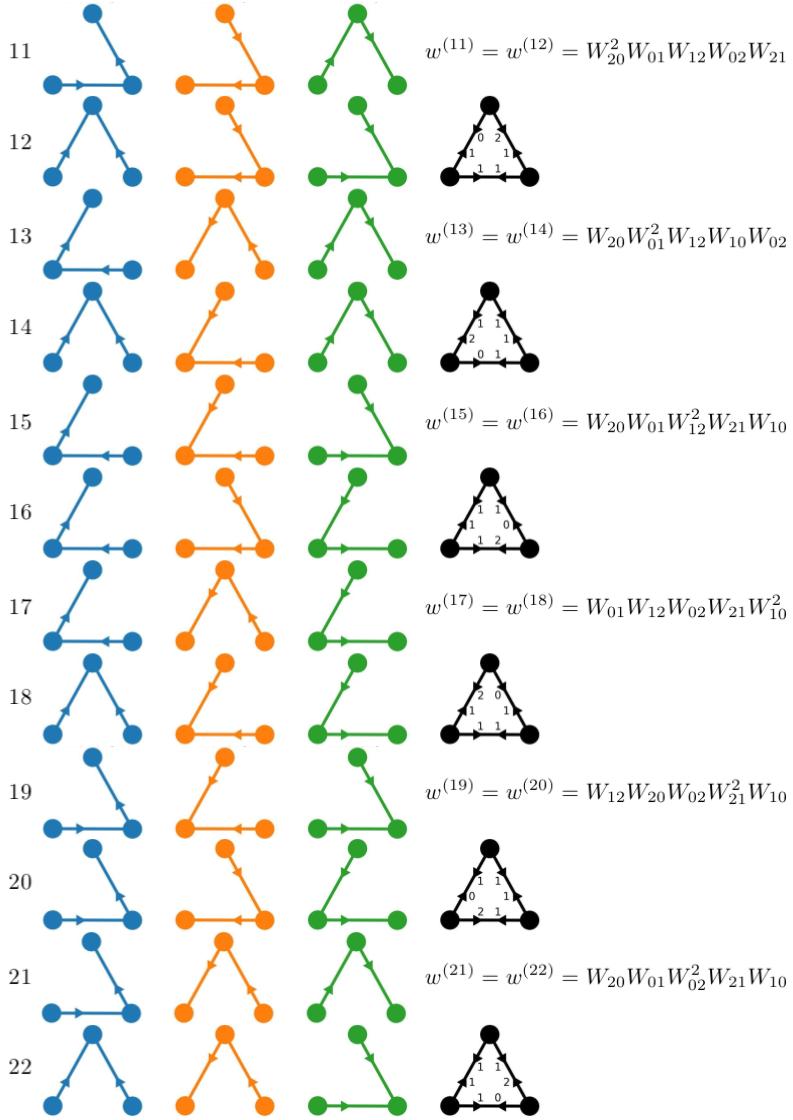}
\caption{The 12 multi-constructable tree families that belong to the 6 equivalency classes of size 2. Identification numbering and tree color coding are the same as in Fig. \ref{treefams_uc}. Every two rows represents a distinct equivalency class. Since the tree families within an equivalency class share the same constituent edges, the weights and edge multiplicity graphs are identical.}
\label{treefams_mc2}
\end{figure}

\begin{figure}
\centering
\includegraphics[width=.95\linewidth]{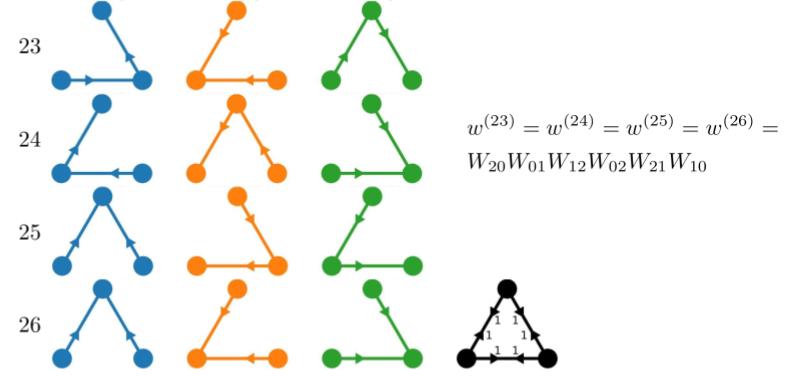}
\caption{The 4 multi-constructable tree families that belong to the single equivalency class of size 4.}
\label{treefams_mc4}
\end{figure}

\begin{figure}
\centering
\includegraphics[width=.95\linewidth]{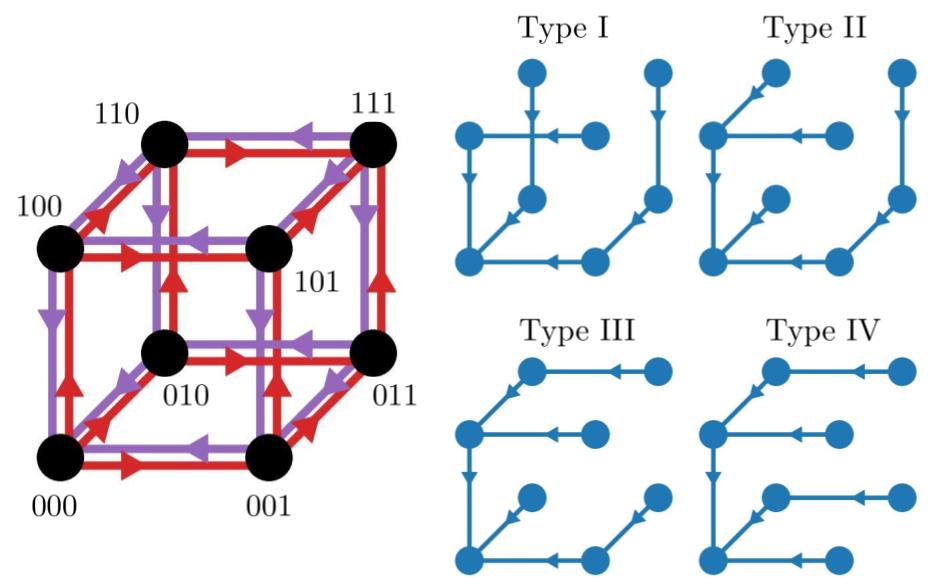}
\caption{Left) General three site binding model with binary labels. Red edges are binding edges while purple edges are unbinding edges. Right) The four types of 000-rooted trees which contain only binding edges. Type I has all double bound states leading into distinct single bound states. Type II has two double bound states leading into the same single bound state and the fully bound state leading into the remaining double bound state. Types III and IV have two double bound states leading into the same single bound state and the fully bound state leading into one of them in a manner perpendicular or parallel to the edge leaving the remaining double bound state. Each type represents 6 different trees obtained through rotations and reflections of the depicted tree.}
\label{unbinds}
\end{figure}

\begin{figure}
\centering
\includegraphics[width=.95\linewidth]{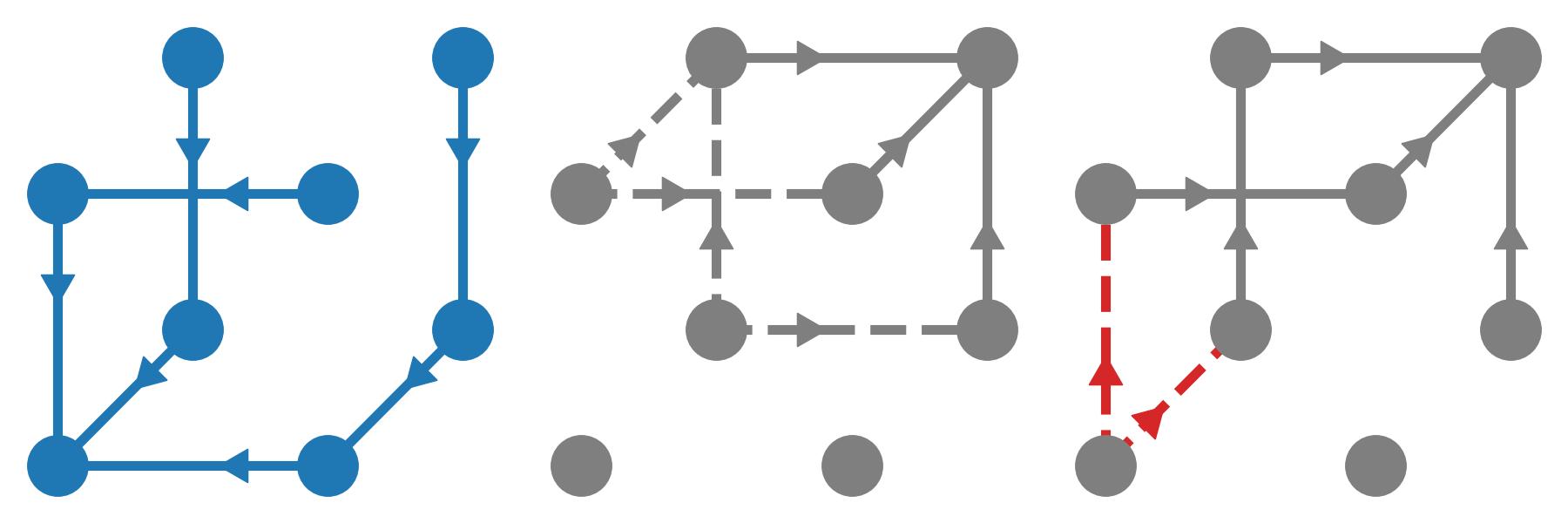}
\caption{Attempted construction of two uniquely constructable trees. Given the Type I 000-rooted tree (blue), we begin constructing a 111-rooted tree (grey). Uniquely constructability dictates which edges the 010 and 100 states must leave by, but this creates an inherent contradiction when attempting to discern which edge the 000 state must leave by.}
\label{type1}
\end{figure}

\begin{figure}
\centering
\includegraphics[width=.63\linewidth]{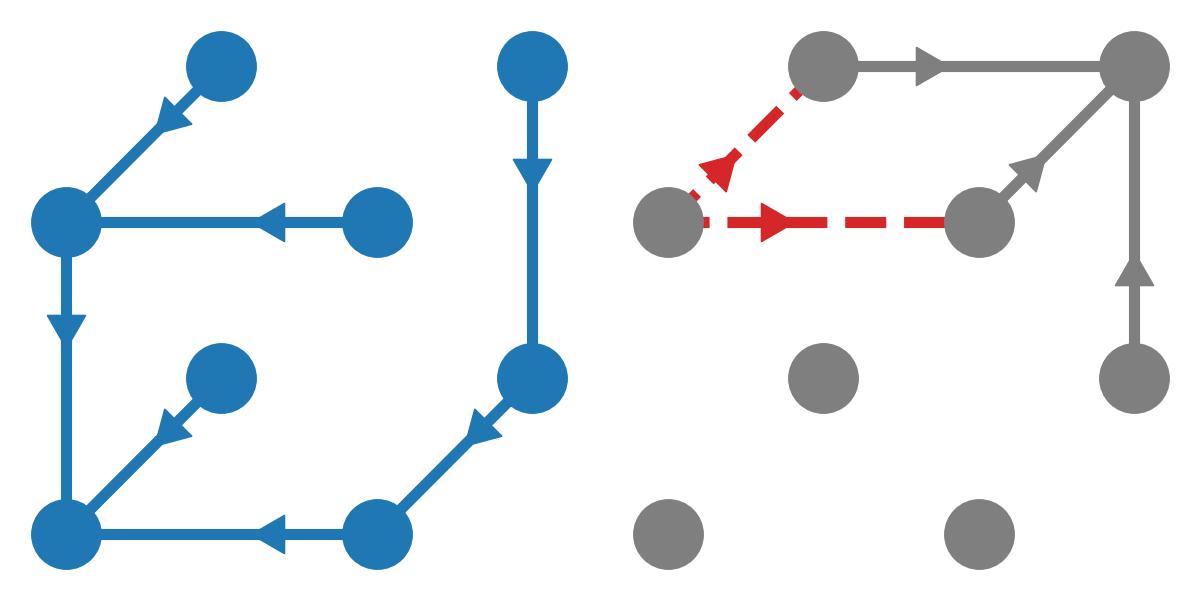}
\caption{Attempted construction of two uniquely constructable trees. Given the Type II 000-rooted tree (blue), we begin constructing a 111-rooted tree (grey). A contradiction is immediately encounted upon examining the possible edges by which the 100 state may leave. Either option will cause the trees to not be uniquely constructable.}
\label{type2}
\end{figure}

\begin{figure}
\centering
\includegraphics[width=.63\linewidth]{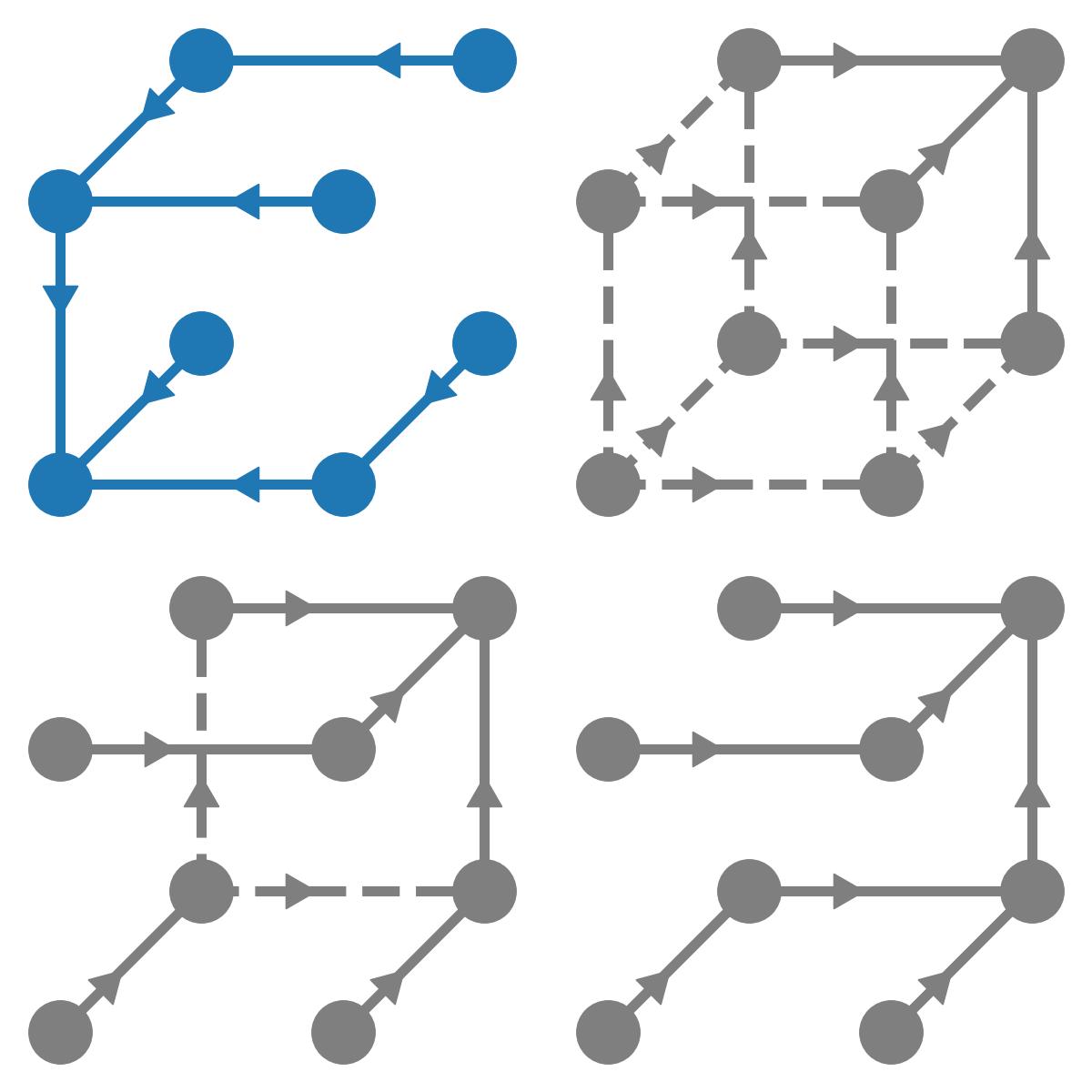}
\caption{Attempted construction of two uniquely constructable trees. Given the Type III 000-rooted tree (blue), we begin constructing a 111-rooted tree (grey). The edges leaving the 000, 001, and 100 state can be deduced so as to prevent exchangability at the 010, 011, and 101 states. The edge leaving the 010 state can then be chosen to prevent simultaneous exchangability at the 001 and 011 states. The resultant trees are uniquely constructable.}
\label{type3}
\end{figure}

\begin{figure}
\centering
\includegraphics[width=.63\linewidth]{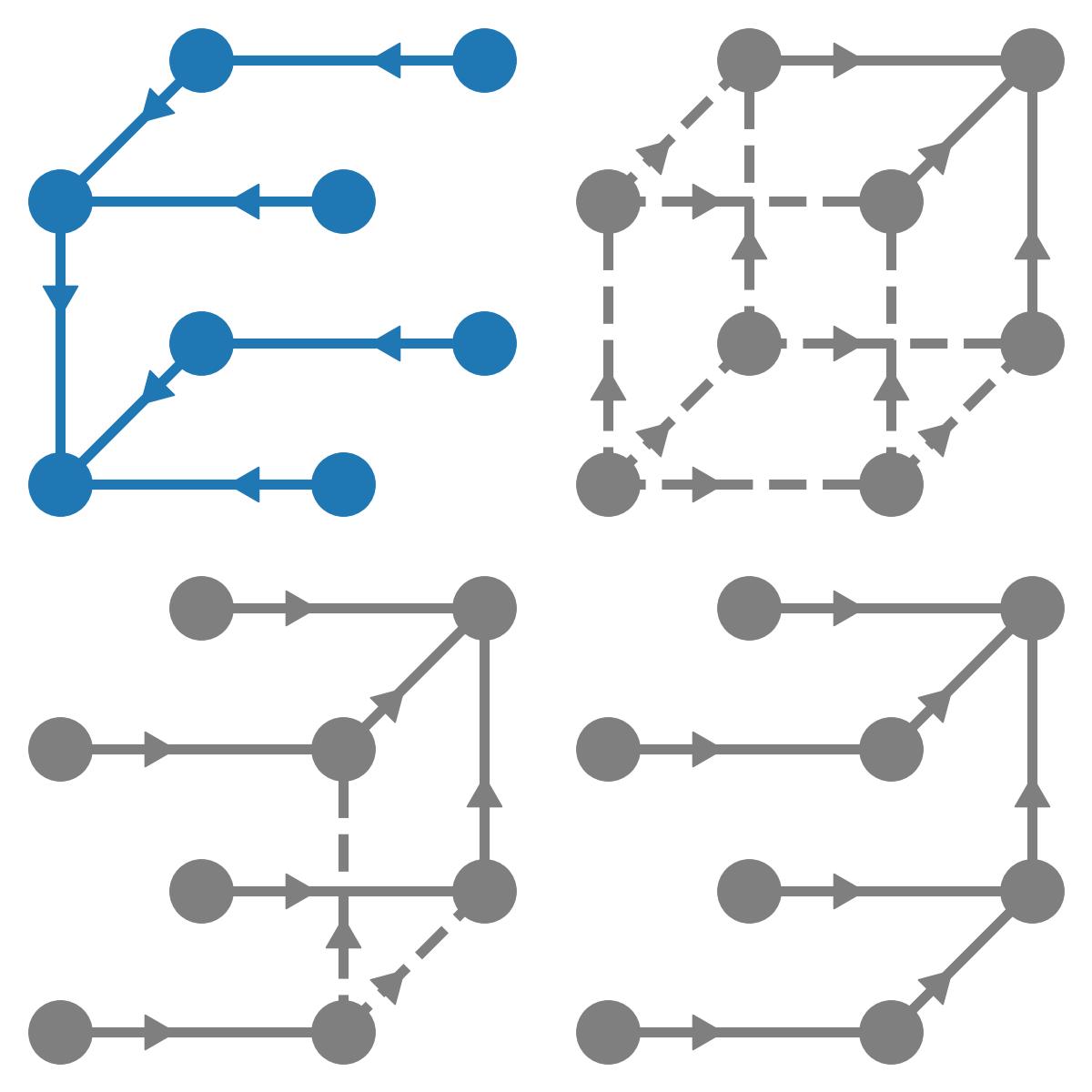}
\caption{Attempted construction of two uniquely constructable trees. Given the Type IV 000-rooted tree (blue), we begin constructing a 111-rooted tree (grey). The edges leaving the 000, 010, and 100 state can be deduced so as to prevent exchangability at the 001, 011, and 101 states. The edge leaving the 001 state can then be chosen to prevent simultaneous exchangability at the 010 and 011 states. The resultant trees are uniquely constructable.}
\label{type4}
\end{figure}

\begin{figure}
\centering
\includegraphics[width=.95\linewidth]{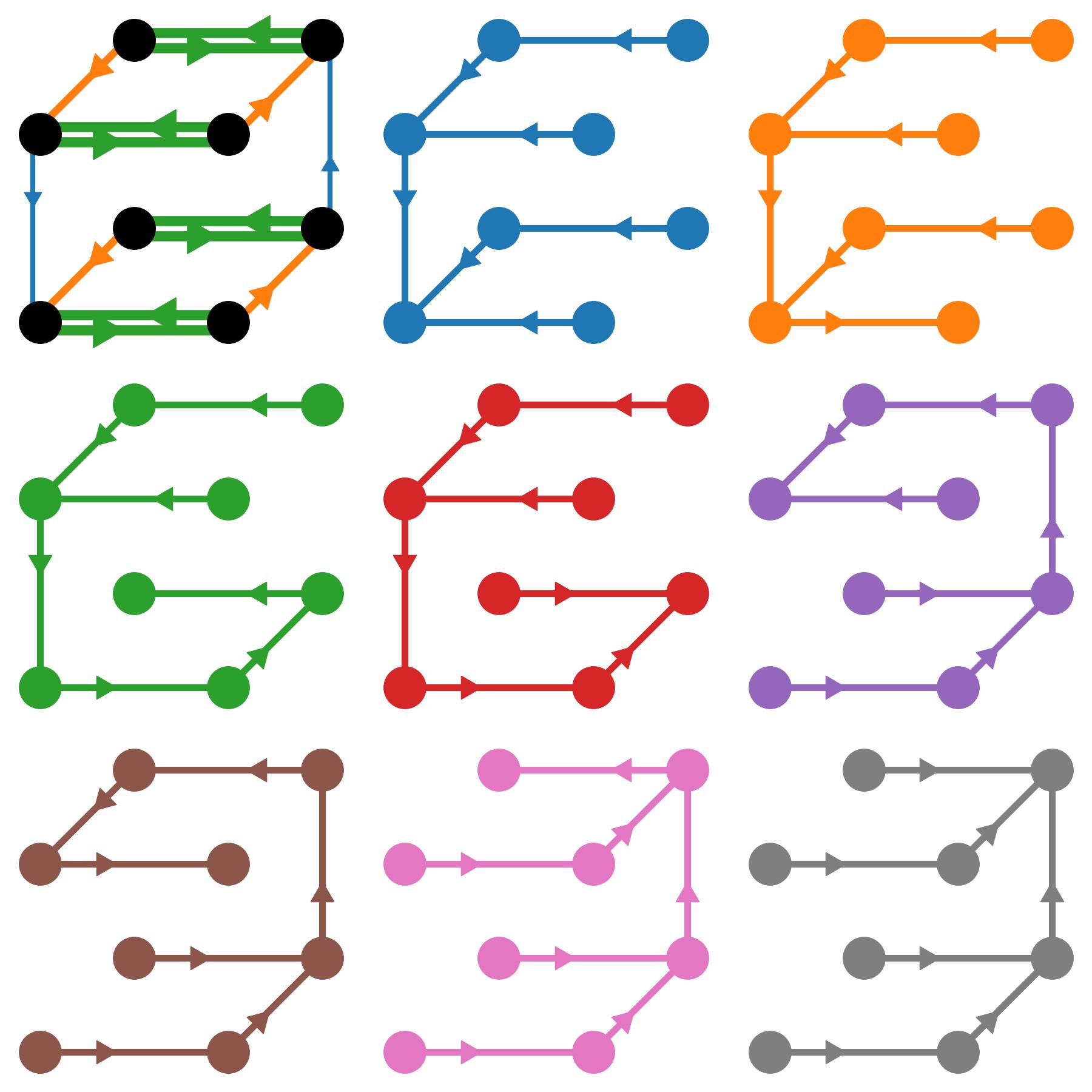}
\caption{A nested hysteresis uniquely constructable tree family.}
\label{3bsvert11}
\end{figure}

\begin{figure}
\centering
\includegraphics[width=.95\linewidth]{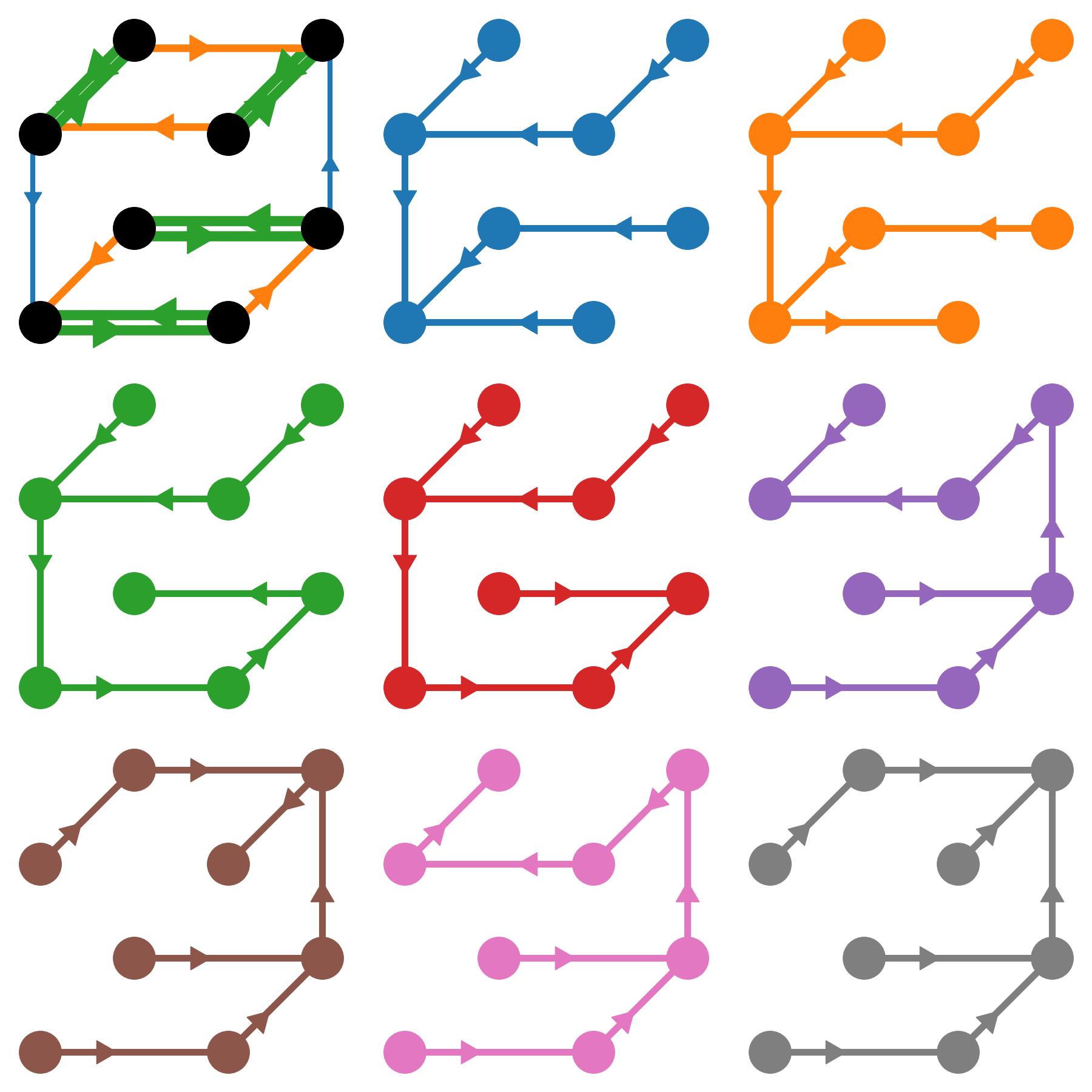}
\caption{A twisted hysteresis uniquely constructable tree family.}
\label{3bsvert07}
\end{figure}

\begin{figure}
\centering
\includegraphics[width=.95\linewidth]{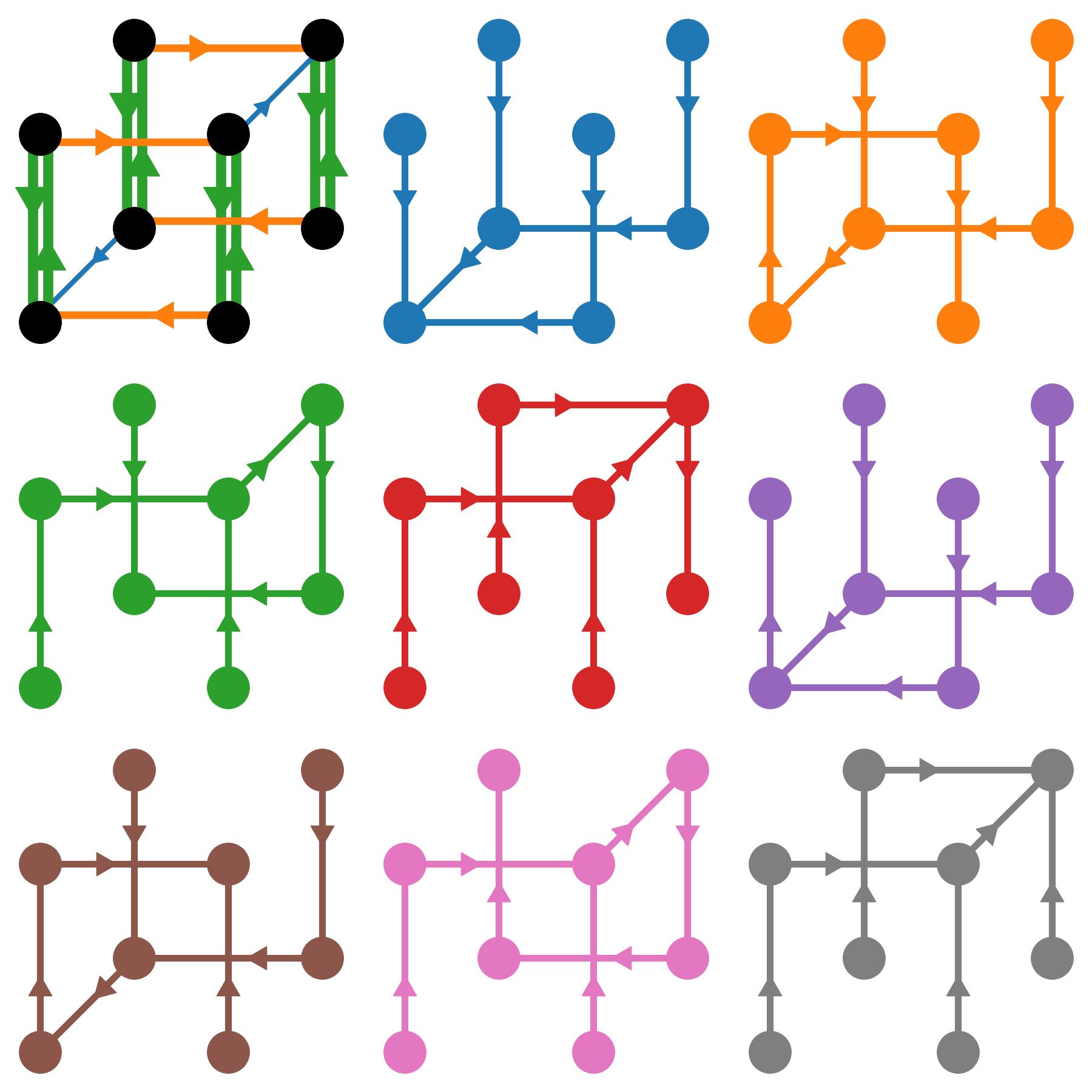}
\caption{A nested hysteresis uniquely constructable tree family.}
\label{3bsvert01}
\end{figure}

\begin{figure}
\centering
\includegraphics[width=.95\linewidth]{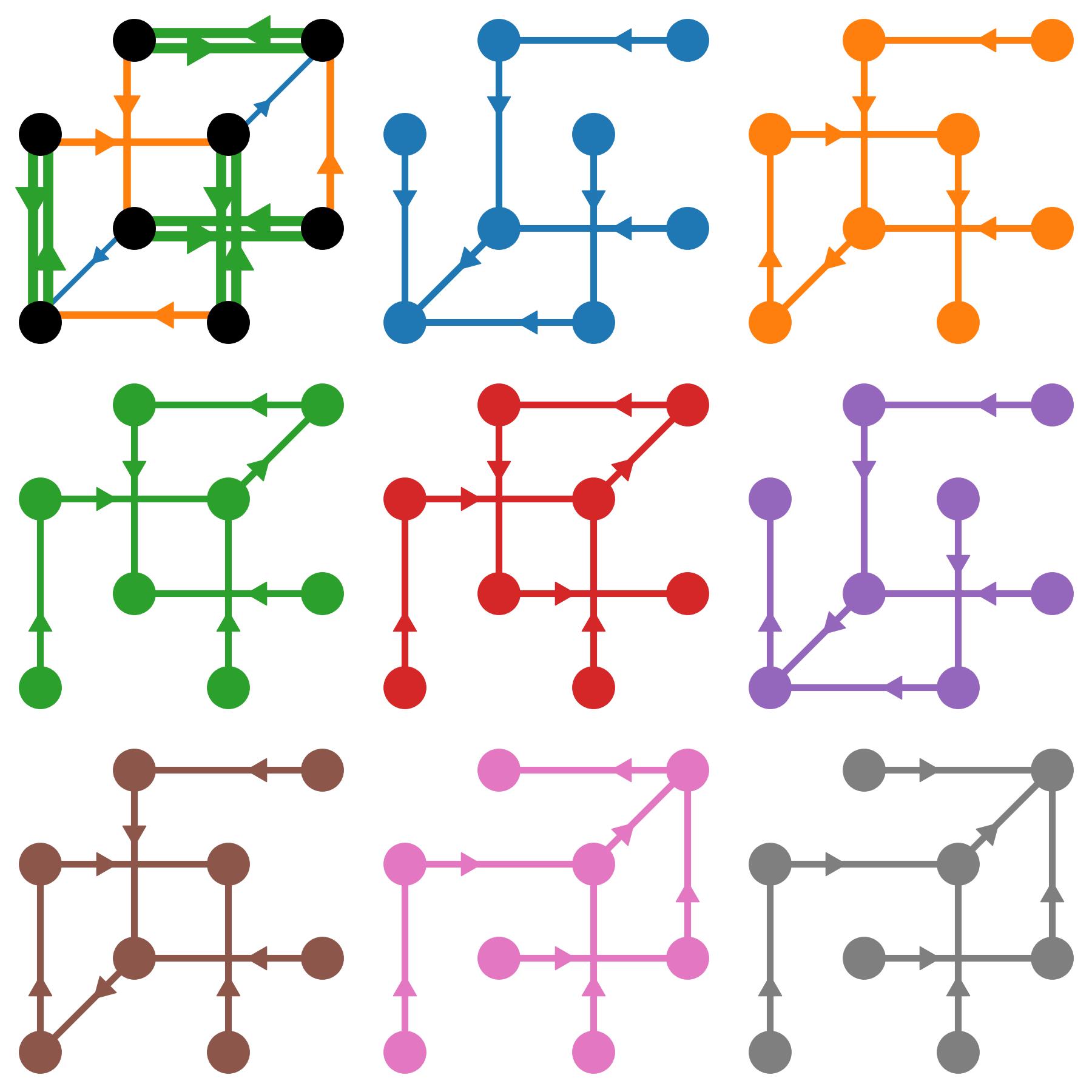}
\caption{A twisted hysteresis uniquely constructable tree family.}
\label{3bsvert08}
\end{figure}

\begin{figure}
\centering
\includegraphics[width=.95\linewidth]{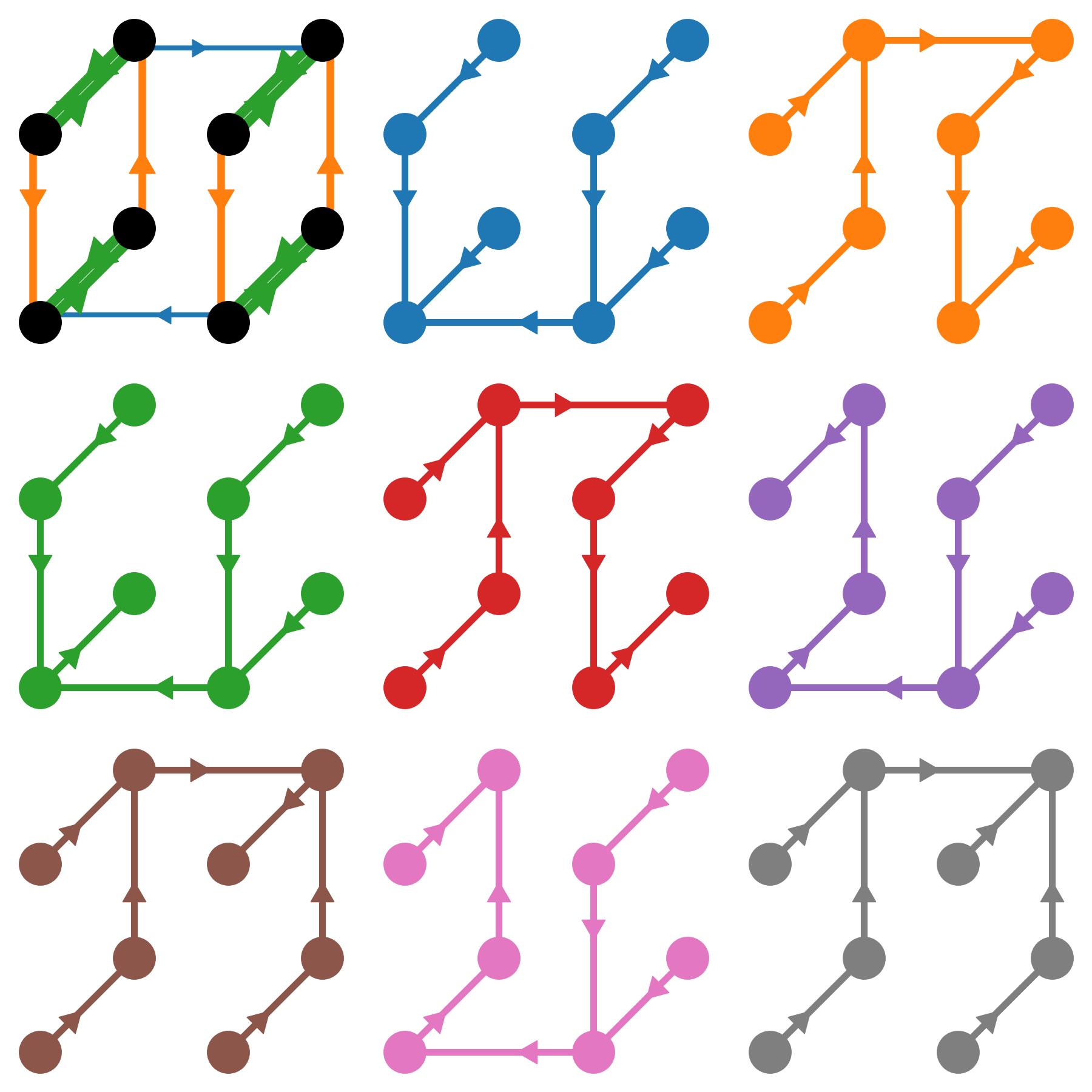}
\caption{A nested hysteresis uniquely constructable tree family.}
\label{3bsvert05}
\end{figure}

\begin{figure}
\centering
\includegraphics[width=.95\linewidth]{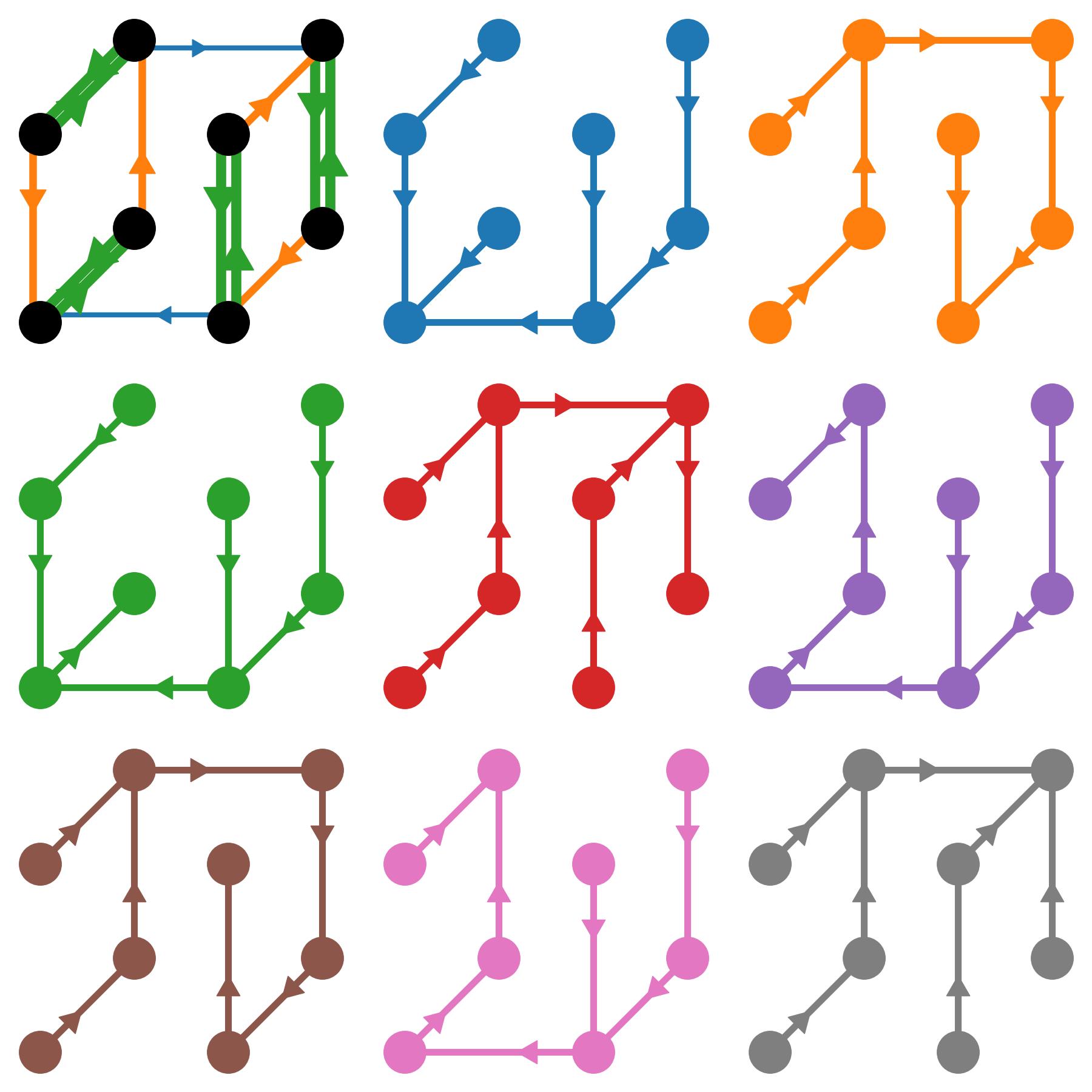}
\caption{A twisted hysteresis uniquely constructable tree family.}
\label{3bsvert03}
\end{figure}

\begin{figure}
\centering
\includegraphics[width=.95\linewidth]{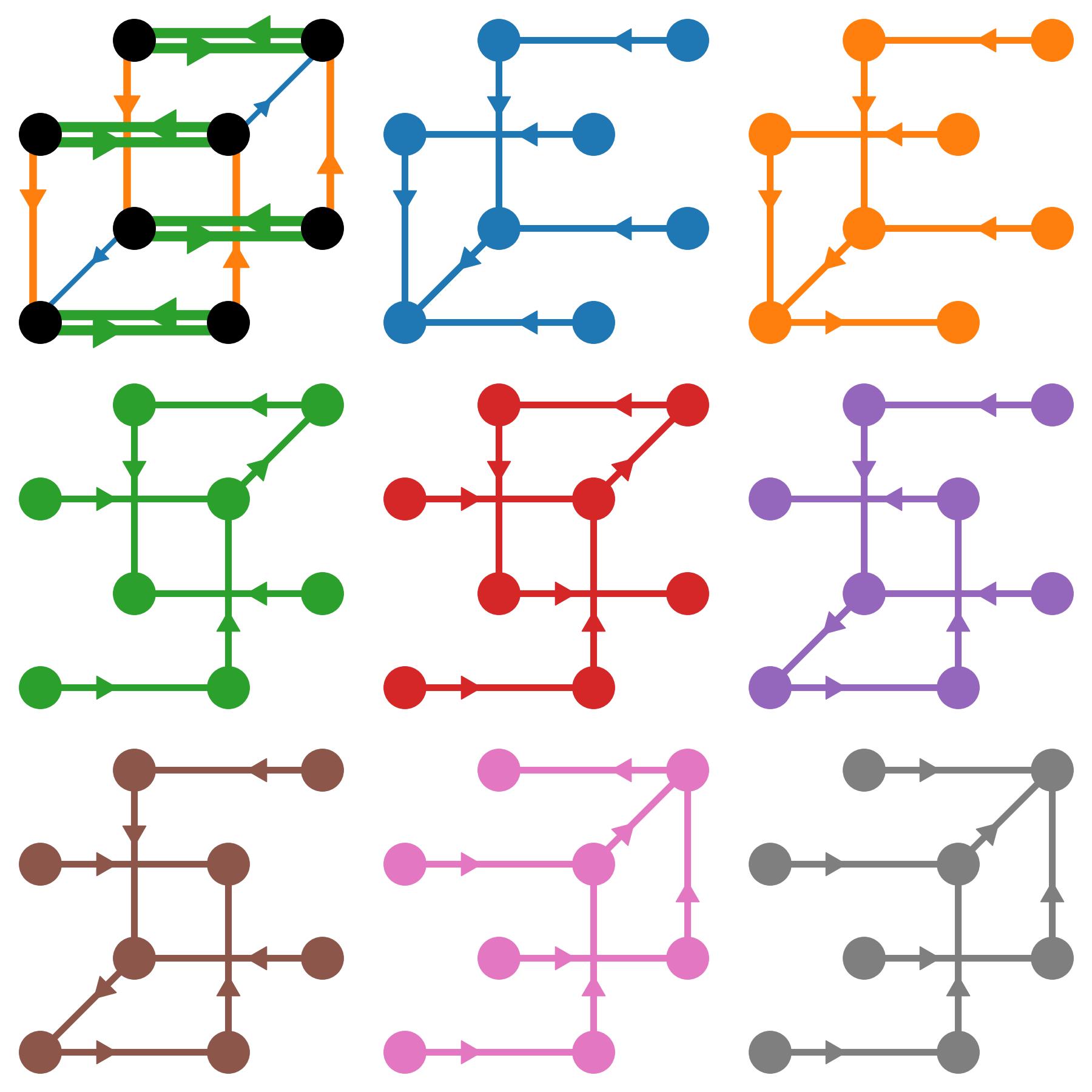}
\caption{A nested hysteresis uniquely constructable tree family.}
\label{3bsvert09}
\end{figure}

\begin{figure}
\centering
\includegraphics[width=.95\linewidth]{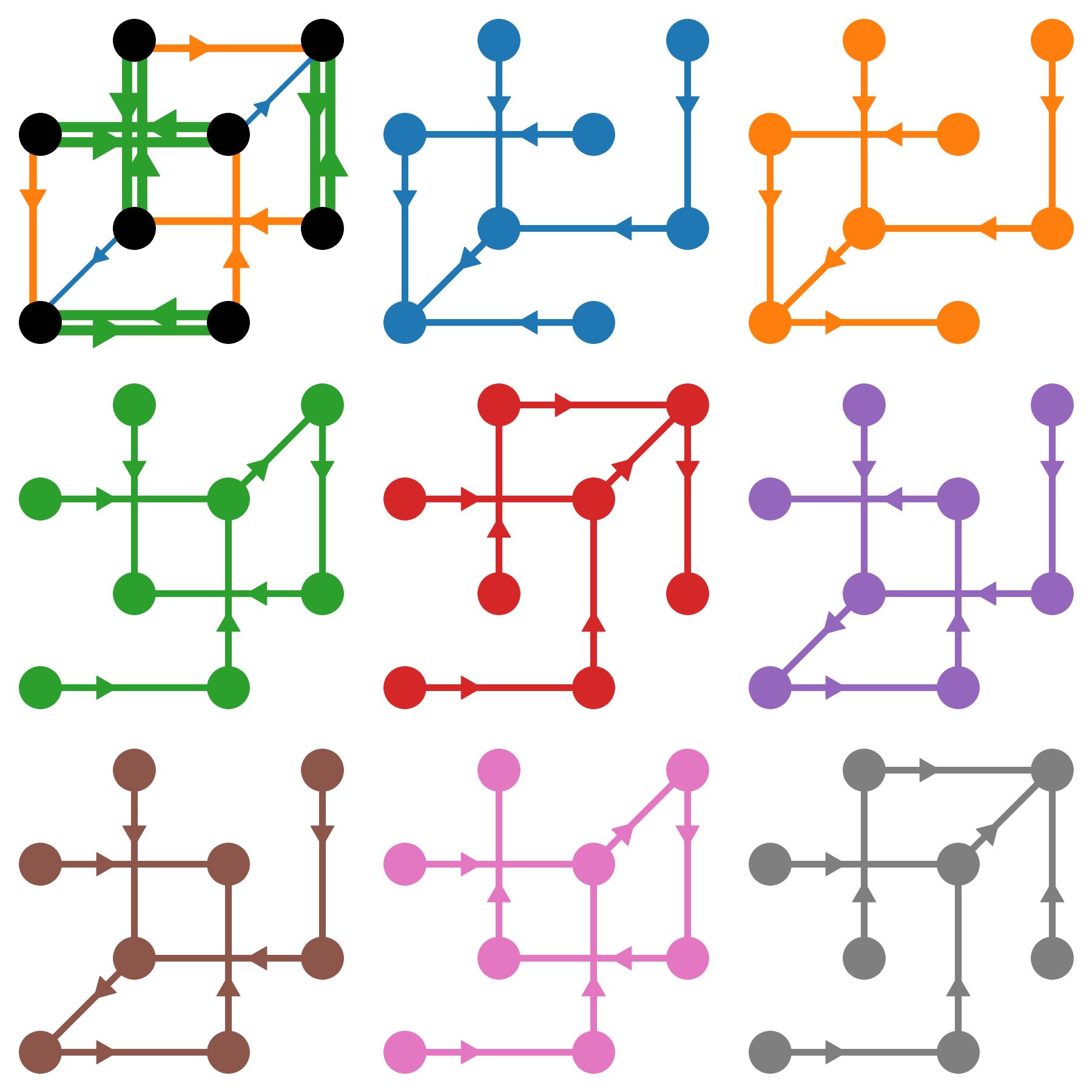}
\caption{A twisted hysteresis uniquely constructable tree family.}
\label{3bsvert02}
\end{figure}

\begin{figure}
\centering
\includegraphics[width=.95\linewidth]{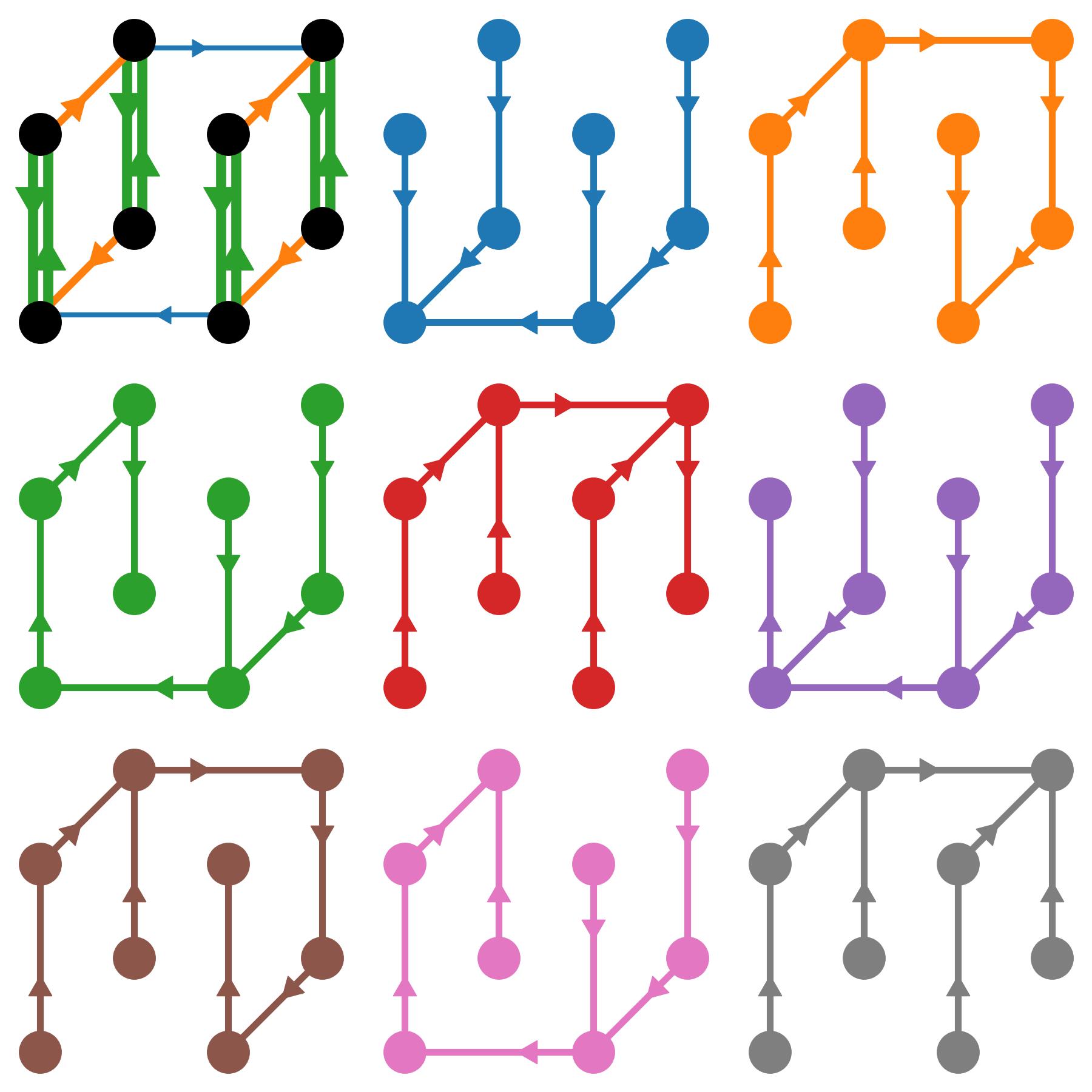}
\caption{A nested hysteresis uniquely constructable tree family.}
\label{3bsvert00}
\end{figure}

\begin{figure}
\centering
\includegraphics[width=.95\linewidth]{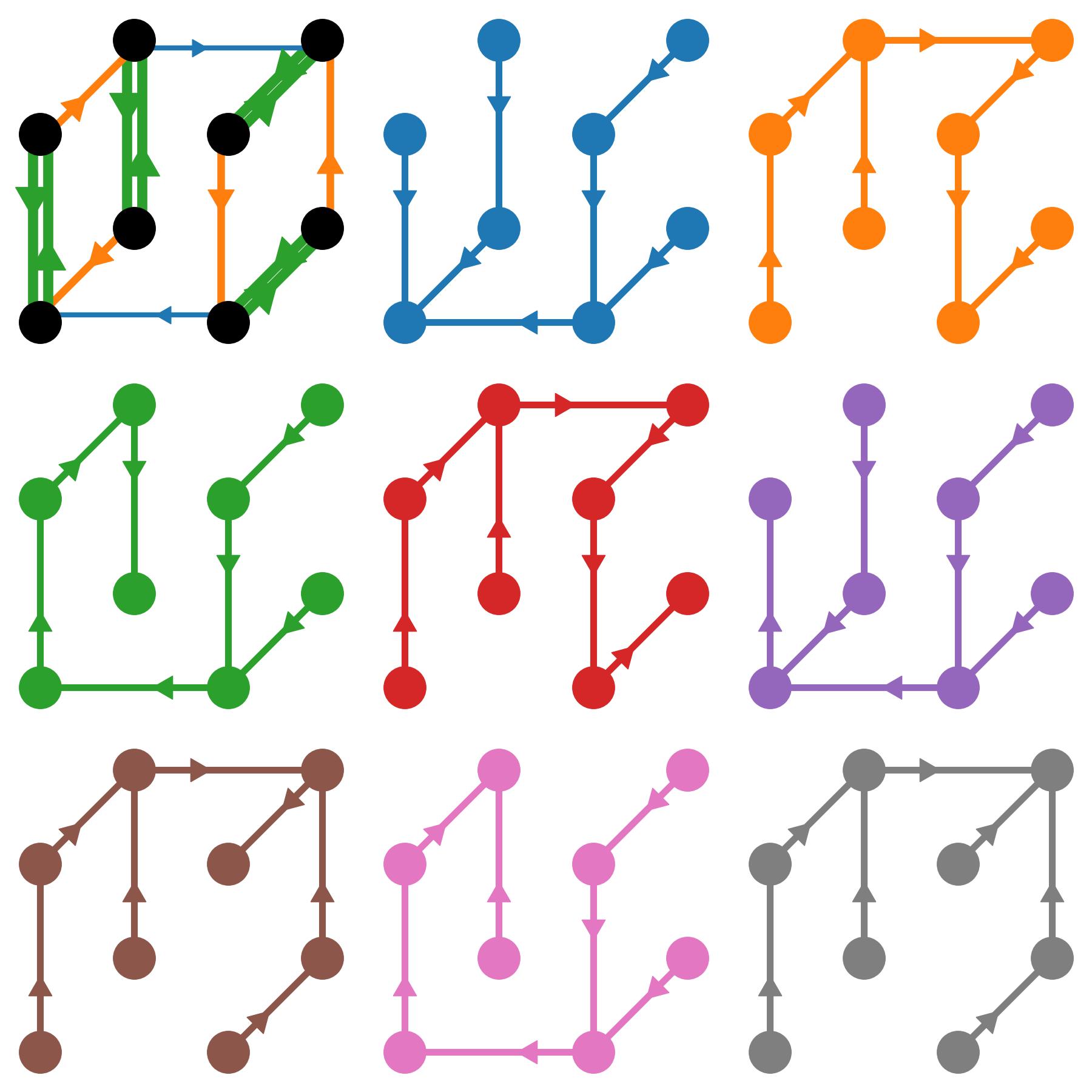}
\caption{A twisted hysteresis uniquely constructable tree family.}
\label{3bsvert04}
\end{figure}

\begin{figure}
\centering
\includegraphics[width=.95\linewidth]{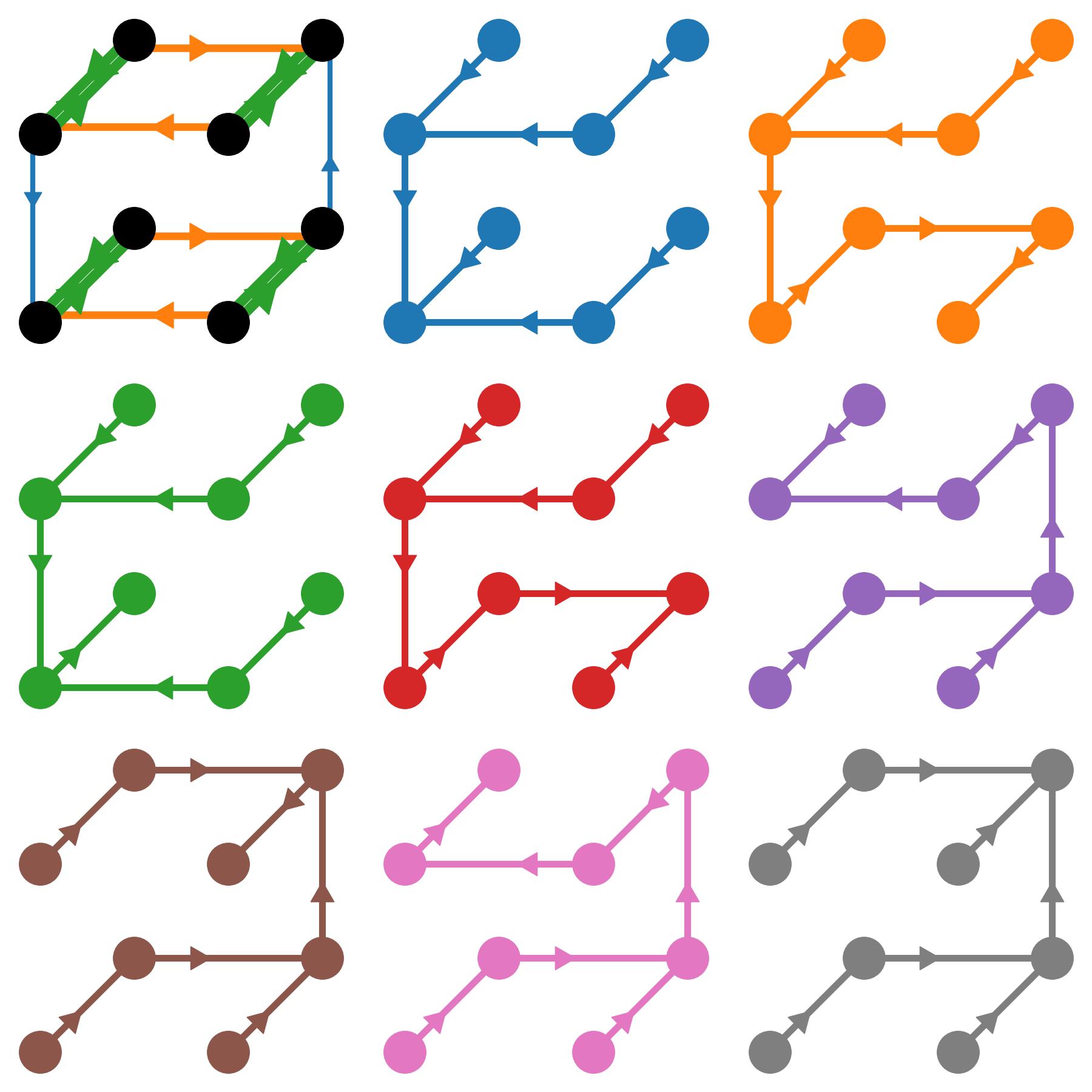}
\caption{A nested hysteresis uniquely constructable tree family.}
\label{3bsvert06}
\end{figure}

\begin{figure}
\centering
\includegraphics[width=.95\linewidth]{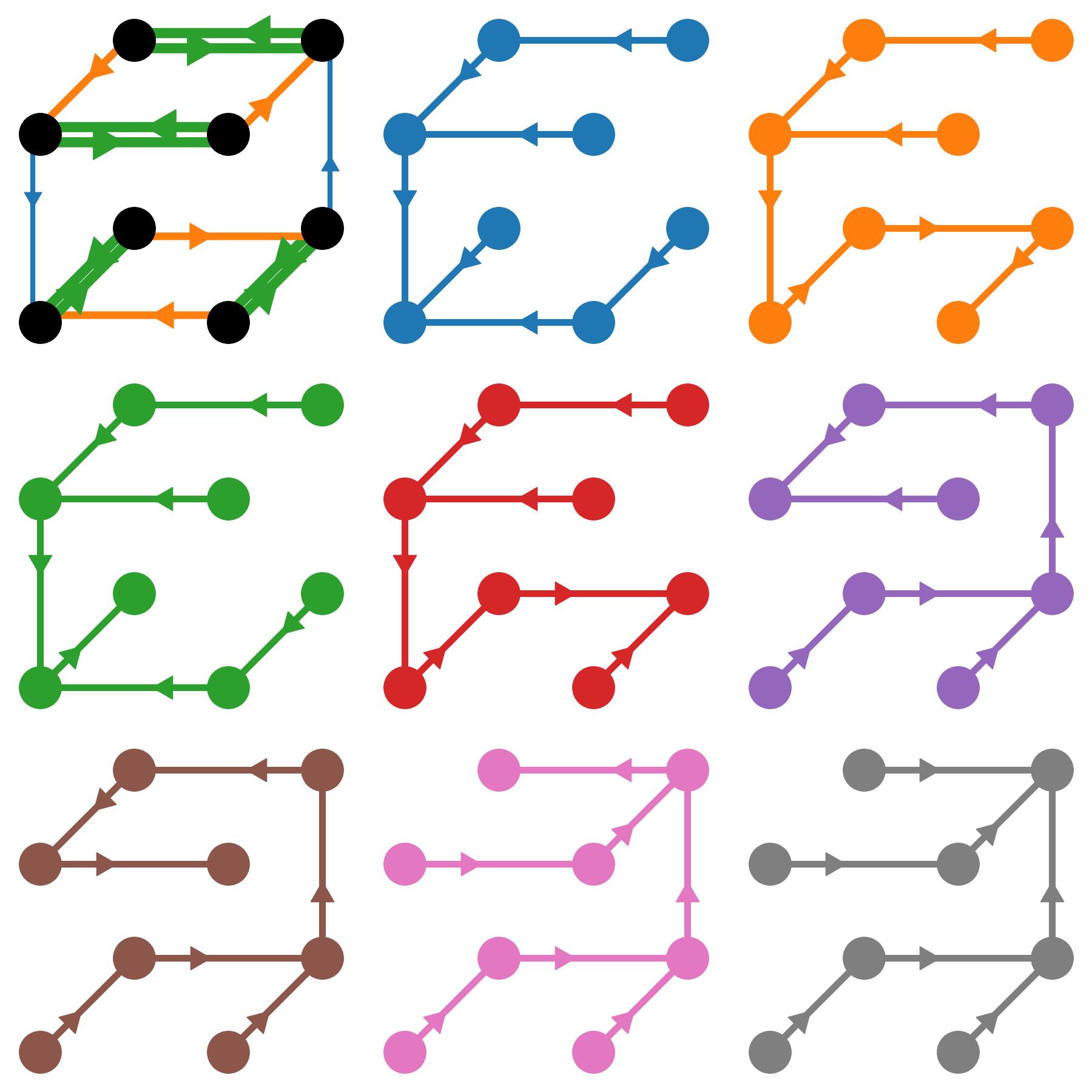}
\caption{A twisted hysteresis uniquely constructable tree family.}
\label{3bsvert10}
\end{figure}

\end{document}